\documentclass[10pt,superscriptaddress,nofootinbib,nobibnotes,tightenlines,twocolumn,
amsmath,amssymb,aps,prl]{revtex4-2}

\usepackage[T1]{fontenc}
\usepackage{amsmath,amssymb,amsthm,mathtools}
\usepackage{microtype}
\usepackage{graphicx}
\usepackage{xcolor}
\usepackage[
  colorlinks=true,
  linkcolor=blue!45!black,
  citecolor=blue!45!black,
  urlcolor=blue!45!black
]{hyperref}

\newtheorem{theorem}{Theorem}
\newtheorem{proposition}[theorem]{Proposition}
\newtheorem{lemma}[theorem]{Lemma}
\newtheorem{corollary}[theorem]{Corollary}
\theoremstyle{remark}
\newtheorem{remark}[theorem]{Remark}

\newcommand{\Tr}{\operatorname{Tr}}
\newcommand{\Herm}{\operatorname{Herm}}
\newcommand{\ran}{\operatorname{ran}}
\newcommand{\SEP}{\operatorname{SEP}}
\newcommand{\FS}{\operatorname{FS}}
\newcommand{\Rg}{R_{\mathrm g}}
\newcommand{\Rs}{R_{\mathrm s}}
\newcommand{\Rob}{R}
\newcommand{\Neg}{\mathcal N}
\newcommand{\id}{I}
\newcommand{\ket}[1]{\lvert #1\rangle}
\newcommand{\bra}[1]{\langle #1\rvert}
\newcommand{\braket}[2]{\langle #1\vert #2\rangle}
\newcommand{\proj}[1]{\ket{#1}\!\bra{#1}}

\begin{document}

\title{The Minimal Dimension of Entangled-Noise Advantage}

\author{Xiao-Ke Wang}
\affiliation{National Laboratory of Solid State Microstructures and School of
Physics, Collaborative Innovation Center of Advanced Microstructures,
Nanjing University, Nanjing 210093, China}
\author{Zi-Yuan Liu}
\affiliation{National Laboratory of Solid State Microstructures and School of
Physics, Collaborative Innovation Center of Advanced Microstructures,
Nanjing University, Nanjing 210093, China}
\author{Ming-Yang Li}
\affiliation{National Laboratory of Solid State Microstructures and School of
Physics, Collaborative Innovation Center of Advanced Microstructures,
Nanjing University, Nanjing 210093, China}
\author{Shengjun Wu}
\email{sjwu@nju.edu.cn}
\affiliation{National Laboratory of Solid State Microstructures and School of
Physics, Collaborative Innovation Center of Advanced Microstructures,
Nanjing University, Nanjing 210093, China}
\author{Zeng-Bing Chen}
\email{zbchen@nju.edu.cn}
\affiliation{National Laboratory of Solid State Microstructures and School of
Physics, Collaborative Innovation Center of Advanced Microstructures,
Nanjing University, Nanjing 210093, China}

\date{September 15, 2026}

\begin{abstract}
Can entangled noise erase entanglement more efficiently than separable noise? Standard robustness restricts the added noise to separable states; generalized robustness allows any state. We prove that the two costs coincide for every two-qubit state and give an explicit full-rank qubit--qutrit state with a strict gap. Positivity under partial transpose (PPT) characterizes separability in both dimensions, so the boundary is not caused by a failure of the PPT criterion. Instead, product-vector geometry permits a rank-one bridge between the two-qubit dual optimizations and supplies a two-dimensional completely entangled subspace for the qubit--qutrit separation. Local isometric embeddings complete the finite-dimensional bipartite classification, also for any fixed multipartite cut. Combined with known three-qubit separation, the result classifies universal equality relative to full separability in all finite multipartite systems.
\end{abstract}

\maketitle

\textit{Introduction.---} Entanglement is a resource for quantum communication and information processing \cite{HorodeckiReview2009,ChitambarGour2019}. Among its measures, robustness asks for the least noise-to-signal ratio needed to make a state separable. Standard robustness restricts the added noise to separable states \cite{VidalTarrach1999}, whereas generalized robustness allows any state \cite{Steiner2003,HarrowNielsen2003}. Their optimizations have witness-dual formulations \cite{Brandao2005}. Separable noise can be prepared locally with shared randomness; entangled noise cannot be prepared with those resources alone. Comparing the two robustnesses asks whether entanglement in the added noise lowers the cost of erasing the input's entanglement. The measures also quantify advantages in distinct discrimination tasks, with different task definitions and normalizations \cite{TakagiRegula2019}. Thus their relation ties the cost of erasure to the resources available for preparing noise.

The two robustnesses agree for bipartite pure states \cite{VidalTarrach1999,Steiner2003}, but mixed states can exhibit a strict separation. Regula \emph{et al.} constructed an infinite-dimensional state with finite generalized and infinite standard robustness \cite{RegulaLamiFerrariTakagi2021,LamiRegulaTakagiFerrari2021}. Lami and Regula subsequently gave an explicit two-qutrit state with \(\Rs=3/4\) and \(\Rg=1/2\), in their analysis of entanglement irreversibility \cite{LamiRegula2023}. These examples show that entangled noise can help, but leave open the smallest bipartite dimension where it does and the dimensions in which equality holds for \emph{every} input. Locating that boundary determines when locally preparable noise always attains the unrestricted optimum and when entangled noise can reduce the cost.

We prove equality for all two-qubit states and give an explicit full-rank \(2\otimes3\) counterexample with an analytic separation certificate. Together these results complete the all-state finite-dimensional classification. Positivity under partial transpose (PPT) still characterizes separability on both sides of this boundary \cite{Peres1996,Horodecki1996}. The low-dimensional exception instead follows from product-vector geometry: every two-dimensional subspace contains a product vector in \(2\otimes2\), but not in \(2\otimes3\) \cite{NiuGriffiths1999,Parthasarathy2004}. Thus the result distinguishes an input-independent absence of noise-resource advantage from its mere absence in special states.

Verstraete and Verschelde obtained its semidefinite program (SDP) and rank-one/filter formulation through optimal teleportation \cite{VerstraeteVerschelde2003}. Their SDP permits arbitrary noise. Within witness duality \cite{Brandao2005}, we adapt their product-vector rank reduction to the \emph{standard} dual and establish the reverse comparison. The appended Supplemental Material (SM), Sec.~\ref{sm:sec:prior-formula-comparison}, gives a counterexample and feasible descent for the earlier generic Wootters-decomposition formula \cite{JafarizadehMirzaeeRezaee2005}, while distinguishing its valid Bell-decomposable special case \cite{AkhtarshenasJafarizadeh2003}.

The preparation question also depends on the party partition. Across a fixed cut, each grouped side may prepare an internally entangled state. With a fully separable target, we instead ask whether local preparation and shared randomness among all parties can attain the unrestricted-noise cost.

\textit{Two noise sets and their duals.---} Let \(\SEP_+\) be the cone of unnormalized separable positive operators. For a bipartite density operator \(\rho\), write the added noise as \(Y=s\omega\), where \(\omega\) is a state and \(s=\Tr Y\). The robustnesses are
\begin{align}
 \Rg(\rho)&=\min\{\Tr Y:Y\succeq0,\ \rho+Y\in\SEP_+\},
 \label{eq:def-rg}\\
 \Rs(\rho)&=\min\{\Tr Y:Y\in\SEP_+,\ \rho+Y\in\SEP_+\}.
 \label{eq:def-rs}
\end{align}
The normalized mixture is \((\rho+Y)/(1+s)\); hence \(s\) is the noise-to-signal ratio. Inclusion of the noise sets gives \(\Rg\le\Rs\).

Physically, unflagged random replacement of a specified input produces \((1-p)\rho+p\omega\), with the replacement event unrecorded. For noise class \(\nu\in\{\mathrm s,\mathrm g\}\), the least probability making the output separable is \(p_\nu^*=R_\nu(\rho)/(1+R_\nu(\rho))\). Standard noise can be prepared locally with shared randomness; generalized noise may require entanglement. Thus equality asks whether this preparation restriction leaves the erasure threshold unchanged. The optimization concerns admixture to a specified state, not arbitrary noise dynamics.

In the two-qubit duals and equality proof below, both local spaces are \(\mathbb C^2\). Let \(\Gamma\) transpose the second subsystem in a fixed product basis and put \(A=\rho^\Gamma\). Since positivity under partial transpose (PPT) is equivalent to separability for two qubits \cite{Peres1996,Horodecki1996}, the generalized constraints are \(Y\succeq0\) and \(A+Y^\Gamma\succeq0\); the standard problem additionally requires \(Y^\Gamma\succeq0\). Their duals are
\begin{align}
 \Rg(\rho)&=\max_{\substack{W\succeq0\\I-W^\Gamma\succeq0}}
                  -\Tr(WA),
 \tag{$D_{\mathrm g}$}\label{eq:Dg}\\
 \Rs(\rho)&=\max_{\substack{W,P,Q\succeq0\\I-W=P+Q^\Gamma}}
                  -\Tr(WA).
 \tag{$D_{\mathrm s}$}\label{eq:Ds}
\end{align}
Strictly feasible primal points \(Y=tI\) exist for sufficiently large \(t\), and trace sublevel sets are compact. Thus both primal and dual optima are attained, with no duality gap \cite{RamanaEtAl1997}; see SM Sec.~\ref{sm:sec:conic}.

\textit{Equality for two qubits.---} An operator \(B\) is block-positive when its expectation is nonnegative on every product vector. Both \(A=\rho^\Gamma\) and the standard-dual slack \(I-W=P+Q^\Gamma\) have this property. In particular,
\begin{equation}
 \bra{e\otimes f}A\ket{e\otimes f}
 =\bra{e\otimes\overline f}\rho
       \ket{e\otimes\overline f}\ge0,
 \label{eq:A-product}
\end{equation}
where the bar denotes complex conjugation in the transpose basis.

\begin{theorem}
\label{thm:main}
For every two-qubit state \(\rho\), \(\Rs(\rho)=\Rg(\rho)\), and the generalized problem admits a separable optimal noise.
\end{theorem}

\begin{proof}
For entangled \(\rho\), choose a standard-dual optimum \((W,P,Q)\). If \(\operatorname{rank}W\ge2\), its range contains a normalized product vector \(\ket p\). With \(W^+\) the Moore--Penrose inverse, set
\begin{equation}
 \begin{gathered}
 t=(\bra pW^+\ket p)^{-1},\\
 W'=W-t\proj p,\qquad P'=P+t\proj p.
 \end{gathered}
 \label{eq:rank-update-main}
\end{equation}
Then \(W',P',Q\succeq0\), \(I-W'=P'+Q^\Gamma\), and the rank decreases. The objective changes by \(t\bra pA\ket p\ge0\), so optimality forces it to remain constant. Iteration yields a rank-one optimum; it cannot be zero because the robustness is positive.

Write it as \(W=w\proj b\), where \(\ket b\) is normalized with Schmidt coefficients \(\alpha_1\ge\alpha_2\ge0\). The largest squared product overlap is \(\alpha_1^2\). Block positivity of \(I-W\) therefore implies \(w\alpha_1^2\le1\). But
\begin{equation}
 \operatorname{spec}((\proj b)^\Gamma)
 =\{\alpha_1^2,\alpha_2^2,\alpha_1\alpha_2,-\alpha_1\alpha_2\},
 \label{eq:rank-one-pt-spectrum}
\end{equation}
so \(I-W^\Gamma\succeq0\). The same optimum is generalized-dual feasible, proving \(\Rs\le\Rg\). Combined with noise-set inclusion, this gives equality. Both values vanish for separable \(\rho\). An attained standard-primal optimum supplies a separable generalized optimum.
\end{proof}

The product-vector lemma and the subtraction step are proved in SM Sec.~\ref{sm:sec:equality-proof}. The proof covers all ranks and requires no nondegenerate spectral decomposition of \(\rho\). It establishes the existence of a common separable optimum, not that every generalized optimizer is separable.

\textit{Optimal-noise preparation.---} For entangled two-qubit states, complementary slackness gives a unique pure-product optimal noise whenever the vector generating a rank-one generalized-dual optimum has unequal Schmidt coefficients. For entangled X states \cite{YuEberly2007,Rau2009}, we determine exactly when one pure-product term suffices and when two are necessary (for fixed-state decomposition lengths, see Ref.~\cite{Wu2004}). The full optimizer certificates, closed-form values and weights, and comparison with the best separable approximation (BSA) \cite{LewensteinSanpera1998,KarnasLewenstein2001} are given in SM Secs.~\ref{sm:sec:variational-certificates}, \ref{sm:sec:x-values}, and~\ref{sm:sec:x-two-product-bsa}.

\textit{Operational corollaries.---} Channel-discrimination theorems identify the optimal success-probability ratio for general ensembles as \(\mathcal A_{\rm gen}=1+\Rg\), and the optimized gain-over-random-guessing ratio for equiprobable binary channels as \(\mathcal A_{\rm bin}=1+2\Rs\) \cite{TakagiRegula2019}. Theorem~\ref{thm:main} therefore gives \(\mathcal A_{\rm bin}=2\mathcal A_{\rm gen}-1\) for two-qubit inputs, although the tasks and their denominators differ. For optimal deterministic two-qubit teleportation, let \(F_{\rm TP\mbox{-}LOCC}^*(\rho)\) be the optimized singlet fraction under trace-preserving local operations and classical communication, and \(f_{\rm tel}^*(\rho)\) the corresponding mean fidelity \cite{VerstraeteVerschelde2003,HorodeckiTeleportation1999}. Their known relation to generalized robustness and Theorem~\ref{thm:main} give, with \(R(\rho)=\Rs(\rho)=\Rg(\rho)\),
\begin{equation}
 \begin{aligned}
 R(\rho)&=2F_{\rm TP\mbox{-}LOCC}^*(\rho)-1\\
        &=3f_{\rm tel}^*(\rho)-2.
 \end{aligned}
 \label{eq:teleportation-common-R-main}
\end{equation}
Our equality identifies this value with the separable-noise cost; task definitions and normalizations are detailed in SM Sec.~\ref{sm:sec:operations}.

\textit{Separation in the minimal higher dimension.---} For \(2\otimes3\), PPT still characterizes separability, so the same SDPs remain exact with \(6\times6\) variables. However, the product-vector step can fail: the span of \(\ket{00}+\ket{11}\) and \(\ket{01}+\ket{12}\) contains no nonzero product vector. Indeed, the \(2\times2\) minors of its coefficient matrix \(\left(\begin{smallmatrix}x&y&0\\0&x&y\end{smallmatrix}\right)\) include \(x^2\) and \(y^2\).

The equality itself fails. In the product basis ordered lexicographically as \(00,01,02,10,11,12\), define
\begin{equation}
 \rho_*=\frac1{40}
 \begin{pmatrix}
 2&0&0&0&0&0\\
 0&1&0&-4&0&0\\
 0&0&13&0&-5&0\\
 0&-4&0&17&0&0\\
 0&0&-5&0&2&0\\
 0&0&0&0&0&5
 \end{pmatrix}.
 \label{eq:2x3-state-main}
\end{equation}
Its trace is one. Up to a basis permutation, \(40\rho_*\) has scalar blocks \(2,5\) and two \(2\times2\) blocks with positive diagonal entries and unit determinants, so \(\rho_*\succ0\). With \(\Gamma\) now transposing the qutrit, set
\begin{equation}
 \ket{u}=2\ket{01}+\ket{10},\qquad
 Y_g=\frac{\proj{u}}{40},\qquad S=\rho_*+Y_g.
 \label{eq:2x3-noise-main}
\end{equation}
The eigenvalues of \(40S^\Gamma\) are \(0,0,4,10,13,18\); hence \(S\) is separable and \(Y_g\) is generalized-primal feasible, with cost \(1/8\). For the standard lower bound, use the subspace above:
\begin{equation}
 \begin{gathered}
 \ket{v_1}=\ket{00}+\ket{11},\qquad
 \ket{v_2}=\ket{01}+\ket{12},\\
 W=\tfrac23(\proj{v_1}+\proj{v_2}),\qquad Z=W^\Gamma .
 \end{gathered}
 \label{eq:2x3-witness-main}
\end{equation}
For every normalized separable state \(\sigma\), \(0\le\Tr(Z\sigma)\le1\); the SM proves this by a \(2\times3\) matrix norm bound, including arbitrary complex product vectors. Thus \(\Rs(\rho_*)\ge-\Tr(Z\rho_*)\), giving
\begin{align}
 \Rg(\rho_*)&\le\frac{1}{8}<\frac{2}{15}\le\Rs(\rho_*),\label{eq:2x3-bounds-main}\\
 \Rs(\rho_*)-\Rg(\rho_*)&\ge\frac{1}{120}>0.
 \label{eq:2x3-gap-main}
\end{align}
These are certified bounds, not asserted exact optimal values. The identities \(WS^\Gamma=0\) and \(\bra{u}Z\ket{u}/\braket{u}{u}=16/15>1\) connect the subspace to the gap: \(\Tr[(Z-I_6)Y_g]=1/120\). The feasible noise \(Y_g/\Tr Y_g=\proj{u}/5\) is entangled; no separable noise can attain this cost.

The separation also persists under white noise. For \(\rho_{1/200}=(199/200)\rho_*+I_6/1200\), the same witness and the scaled noise \((199/200)Y_g\) give
\begin{equation}
 \Rs(\rho_{1/200})-\Rg(\rho_{1/200})
 \ge\frac{437}{72000}>0.00606.
 \label{eq:fullrank-gap-main}
\end{equation}
\textit{Universal boundary and multipartite extensions.---} Both robustnesses are invariant under local isometric embeddings: a local trace-preserving retraction recovers the input and preserves either admissible noise set (SM Sec.~\ref{sm:sec:dimension}). Every bipartite space with \(m,n\ge2\), other than \(2\otimes2\), contains \(2\otimes3\) up to exchanging the subsystems. We therefore obtain the complete all-state classification
\begin{equation}
 \begin{gathered}
 \Rs(\rho)=\Rg(\rho)\ \text{for every }\rho
 \text{ on }\mathbb C^m\otimes\mathbb C^n\\
 \Longleftrightarrow\quad(m,n)=(2,2),\qquad m,n\ge2.
 \end{gathered}
 \label{eq:dimension-classification}
\end{equation}
Here \(m,n\) are finite integers; if \(\min\{m,n\}=1\), both values vanish for all states. Higher-dimensional spaces contain separating states, not exclusively such states; bipartite pure-state equality remains valid. Total dimension six is the smallest possible separation, already before the PPT criterion ceases to characterize separability.

For a finite multipartite system \(\bigotimes_{j=1}^N\mathbb C^{d_j}\), the same classification applies to any fixed nontrivial bipartition \(S|\bar S\), with effective dimensions \(d_S=\prod_{i\in S}d_i\) and \(d_{\bar S}=\prod_{i\notin S}d_i\). Separability across this cut allows entanglement within each side. Full separability instead uses the free set \cite{WuChenZhang2000}
\begin{equation}
 \FS_N=\operatorname{conv}\left\{
 \bigotimes_{j=1}^N\proj{x_j}:\ket{x_j}\in\mathbb C^{d_j},
 \ \braket{x_j}{x_j}=1\right\}.
 \label{eq:main-fs-definition}
\end{equation}
Write \(R_\nu^{\FS_N}\), \(\nu\in\{\mathrm s,\mathrm g\}\), for the corresponding robustnesses. Delete one-dimensional parties and relabel the remaining dimensions \(d_1,\ldots,d_k\). The bipartite classification, known three-qubit separation \cite{Contreras2019}, and invariance under local embeddings and separable spectators imply (SM Sec.~\ref{sm:sec:fs-multipartite})
\begin{equation}
 \begin{aligned}
 R_{\mathrm s}^{\FS_N}(\rho)&=R_{\mathrm g}^{\FS_N}(\rho)\quad\text{for every }\rho\\[-2pt]
 &\Longleftrightarrow\quad k\le1\\[-2pt]
 &\qquad\text{or }\bigl(k=2,\ d_1=d_2=2\bigr).
 \end{aligned}
 \label{eq:main-fs-classification}
\end{equation}
For a quantitative comparison, define the Greenberger--Horne--Zeilinger (GHZ) states \(G_n=\proj{g_n}\), where \(\ket{g_n}=(\ket{0}^{\otimes n}+\ket{1}^{\otimes n})/\sqrt2\), for \(n\ge2\) \cite{WuZhang2000}. SM Sec.~\ref{sm:sec:fs-multipartite} gives matching witnesses and fully separable decompositions for
\begin{equation}
 R_{\mathrm g}^{\FS_n}(G_n)=1,
 \qquad R_{\mathrm s}^{\FS_n}(G_n)=2^{n-2}.
 \label{eq:main-ghz-values}
\end{equation}
The generalized value follows from known stabilizer-state results \cite{Hayashi2008}, and the \(n=3\) separation was established in Ref.~\cite{Contreras2019}. Across any fixed bipartition, the same GHZ state has both bipartite robustnesses equal to one. Thus the party partition and free set determine the advantage: the exponential factor concerns the optimized noise-to-signal ratio, not a noise probability.

Verification scripts for the rational certificates and GHZ identities are available in Ref.~\cite{VerificationCode}.

\textit{Conclusion.---} Entangled noise has no cost advantage for any two-qubit input, yet an explicit full-rank qubit--qutrit state already exhibits a strict advantage. PPT remains exact in both settings; the change is exposed by product-vector geometry, not by the onset of PPT entanglement. Together with local embeddings and known three-qubit separation, these results locate the universal-equality boundary for bipartite systems and for multipartite systems whose free states are fully separable. This classification treats exact, single-copy robustness and does not address many-copy or asymptotic tasks.

\begin{acknowledgments}
This work is supported by the National Natural Science Foundation of China (Grants 12475020 and 92565111), Quantum Science and Technology-National Science and Technology Major Project (2021ZD0301701), and the National Key Research and Development Program of China (2023YFC2205802).
\end{acknowledgments}

\makeatletter
\let\auto@bib@innerbib\@empty
\makeatother

\clearpage
\onecolumngrid
\rightskip=0pt plus 0pt
\setcounter{section}{0}
\setcounter{subsection}{0}
\setcounter{subsubsection}{0}
\setcounter{equation}{0}
\setcounter{figure}{0}
\setcounter{table}{0}
\setcounter{theorem}{0}
\setcounter{secnumdepth}{2}
\renewcommand{\thefigure}{S\arabic{figure}}
\renewcommand{\thetable}{S\arabic{table}}
\renewcommand{\theequation}{S\arabic{equation}}
\renewcommand{\thesection}{\Roman{section}}
\renewcommand{\thesubsection}{\Alph{subsection}}
\renewcommand{\thetheorem}{\thesection.\arabic{theorem}}
\renewcommand{\theHsection}{sm.\arabic{section}}
\renewcommand{\theHsubsection}{sm.\arabic{section}.\arabic{subsection}}
\renewcommand{\theHequation}{sm.\arabic{equation}}
\renewcommand{\theHfigure}{sm.\arabic{figure}}
\renewcommand{\theHtable}{sm.\arabic{table}}
\renewcommand{\theHtheorem}{sm.\arabic{theorem}}

\begin{center}
{\Large\bfseries Supplemental Material for\\[0.5ex]
The Minimal Dimension of Entangled-Noise Advantage\par}
\vspace{1.2ex}
{\large Xiao-Ke Wang, Zi-Yuan Liu, Ming-Yang Li, Shengjun Wu, and Zeng-Bing Chen\par}
\vspace{0.6ex}
{\small National Laboratory of Solid State Microstructures and School of Physics,\\
Collaborative Innovation Center of Advanced Microstructures,\\
Nanjing University, Nanjing 210093, China\par}
\end{center}
\vspace{1.5ex}

\section*{Guide to results and proofs}
The material is organized by the role of the arguments in the main results.
\begin{itemize}
 \item \emph{Core equality:} strict feasibility and attainment in Sec.~\ref{sm:sec:conic}; product-subspace and rank-subtraction lemmas, standard-dual rank reduction (Proposition~\ref{sm:prop:rank-one-standard}), normalization bridge (Proposition~\ref{sm:prop:rank-one-bridge}), and all-state equality (Theorem~\ref{sm:thm:equality}) in Sec.~\ref{sm:sec:equality-proof}.
 \item \emph{Optimizer structure:} the single-vector formula (equivalent to the known generalized filter formula), complementary slackness, conditional unique product noise, the Bell-state alternatives, the nonlinear kernel certificate, and inverse partial-transpose reduction are in Sec.~\ref{sm:sec:variational-certificates}.
 \item \emph{X states:} Theorem~\ref{sm:thm:x-state} gives the values for all ranks and boundaries in Sec.~\ref{sm:sec:x-values}. Section~\ref{sm:sec:x-two-product-bsa} gives the two-product compression, minimum preparation size, and BSA comparison.
 \item \emph{Nearest prior work:} Sec.~\ref{sm:sec:prior-formula-comparison} contains the Wootters specialization, full-rank counterexample and feasible descent, the Verstraete--Verschelde comparison, and the 2021/2023 dimension context.
 \item \emph{Bounds and applications:} Secs.~\ref{sm:sec:bounds} and \ref{sm:sec:operations} retain negativity bounds with exact upper saturation and lower strictness, and the distinct operational normalizations as corollaries of established task theorems.
 \item \emph{Dimension boundary:} Sec.~\ref{sm:sec:dimension} supplies all rational matrices and positivity certificates, full-rank stability, local isometry invariance, the complete bipartite classification, and its fixed-cut extension (Corollary~\ref{sm:cor:fixed-bipartition}).
 \item \emph{Full separability:} Sec.~\ref{sm:sec:fs-multipartite} gives the finite multipartite FS definition, spectator and embedding invariance, the all-\(n\) GHZ witnesses and values, the full-rank white-noise family, and the universal-equality classification.
\end{itemize}

\section{Conic and semidefinite formulations}\label{sm:sec:conic}
\subsection{States, separability, and partial transpose}
Unless stated otherwise, \(\mathcal H_A=\mathcal H_B=\mathbb C^2\). Write \(X\succeq0\) for positive semidefiniteness and \(X\succ0\) for positive definiteness; \(\Herm(d)\) is the real space of Hermitian \(d\times d\) matrices. The cone of unnormalized separable operators and its trace-one slice are
\begin{align}
 \SEP_+&=\left\{\sum_{r=1}^N A_r\otimes B_r:
             A_r,B_r\succeq0\right\},\label{sm:eq:sepcone}\\
 \SEP&=\{\sigma\in\SEP_+:\Tr\sigma=1\}.
 \label{sm:eq:separable-state-set}
\end{align}
We use \(\mathcal D(\mathcal H)=\{X\succeq0:\Tr X=1\}\) for the set of
density operators on a finite-dimensional Hilbert space.
In a fixed product basis, partial transpose on \(B\) is defined by
\begin{equation}
 (X^\Gamma)_{ij,k\ell}=X_{i\ell,kj}.
 \label{sm:eq:partial-transpose}
\end{equation}
It is a trace-preserving, self-adjoint involution:
\begin{equation}
 (X^\Gamma)^\Gamma=X,\quad \Tr X^\Gamma=\Tr X,\quad
 \Tr(XY^\Gamma)=\Tr(X^\Gamma Y).
 \label{sm:eq:pt-identities}
\end{equation}
The PPT criterion in \(2\otimes2\) is
\begin{equation}
 X\in\SEP_+\quad\Longleftrightarrow\quad
 X\succeq0,\ X^\Gamma\succeq0
 \label{sm:eq:ppt-sep}
\end{equation}
\cite{Peres1996,Horodecki1996,WuAnandan2002}. Both positivity conditions are retained throughout.

\subsection{Block positivity and decomposable operators}
An operator \(D\in\Herm(4)\) is block-positive if
\begin{equation}
 \bra{a\otimes b}D\ket{a\otimes b}\ge0
 \quad\text{for all }\ket a,\ket b.
 \label{sm:eq:block-positive-definition}
\end{equation}
This is equivalent to membership in the dual cone
\(\SEP_+^*=\{D:\Tr(DX)\ge0\ \forall X\in\SEP_+\}\).
Every decomposable operator
\begin{equation}
 D=P+Q^\Gamma,\qquad P,Q\succeq0
 \label{sm:eq:decomposable}
\end{equation}
is block-positive \cite{LewensteinEtAl2000,ChruscinskiSarbicki2014}. The variable \(W\succeq0\) below is a dual operator; when \(-\Tr(W\rho^\Gamma)>0\), its partial transpose \(W^\Gamma\) is an entanglement witness.

\subsection{The two robustnesses}
For a density operator \(\rho\), define \cite{VidalTarrach1999,Steiner2003,HarrowNielsen2003}
\begin{align}
 \Rs(\rho)&=\inf\{s\ge0:\omega\in\SEP,\
            (\rho+s\omega)/(1+s)\in\SEP\},\label{sm:eq:def-rs}\\
 \Rg(\rho)&=\inf\{s\ge0:\omega\succeq0,\ \Tr\omega=1,\
            (\rho+s\omega)/(1+s)\in\SEP\}.\label{sm:eq:def-rg}
\end{align}
Here \(s\) is a noise-to-signal ratio; the probability of adding noise in the normalized mixture is \(s/(1+s)\). Inclusion of the admissible noise sets gives
\begin{equation}
 \Rg(\rho)\le\Rs(\rho).
 \label{sm:eq:easy-direction}
\end{equation}

\subsection{Primal and dual semidefinite programs}
Set
\begin{equation}
 A:=\rho^\Gamma.
 \label{sm:eq:A-definition}
\end{equation}
For generalized robustness, write the unnormalized noise as \(Y=s\omega\).  Then \(Y\succeq0\) and \(\Tr Y=s\).  The operator \(\rho+Y\) is already positive semidefinite; by \eqref{sm:eq:ppt-sep}, it is separable precisely when \(A+Y^\Gamma\succeq0\).  Hence
\begin{equation*}
 \Rg(\rho)=
 \min_{Y\in\Herm(4)}
 \left\{\Tr Y:Y\succeq0,\ A+Y^\Gamma\succeq0\right\}.
 \tag{$P_{\mathrm g}$}\label{sm:eq:Pg}
\end{equation*}
This is the generalized-robustness SDP used, with an equivalent normalization, in Ref.~\cite{VerstraeteVerschelde2003}. Conversely, a feasible \(Y\ne0\) determines the noise state \(\omega=Y/\Tr Y\); the feasible point \(Y=0\) corresponds to \(s=0\).

For standard robustness the unnormalized noise \(Y\) must be separable. Put \(K=Y^\Gamma\).  The two-qubit PPT criterion, the involution property, and trace preservation give
\begin{equation*}
 \Rs(\rho)=
 \min_{K\in\Herm(4)}
 \left\{\begin{aligned}
 &\Tr K:\ K\succeq0,\ K^\Gamma\succeq0,\\[-2pt]
 &A+K\succeq0
 \end{aligned}\right\}.
 \tag{$P_{\mathrm s}$}\label{sm:eq:Ps}
\end{equation*}
Here \(K^\Gamma=Y\succeq0\) and \(K=Y^\Gamma\succeq0\) are both retained; no implication from PPT to physical positivity has been assumed. The zero update is feasible exactly for separable states.

Both displayed minima are attained.  Indeed, a positive semidefinite variable has trace norm equal to its trace.  Every finite objective sublevel set is therefore bounded, and its intersection with the closed feasible set is compact.

We derive the duals using the Hilbert--Schmidt inner product \(\langle X,Z\rangle=\Tr(XZ)\) on \(\Herm(4)\).  For \((P_{\mathrm g})\), introduce \(W\succeq0\) for the constraint \(A+Y^\Gamma\succeq0\), while keeping \(Y\succeq0\) as the variable domain. By \eqref{sm:eq:pt-identities},
\begin{align}
 L_{\mathrm g}(Y,W)
 &=\Tr Y-\Tr\!\left[W(A+Y^\Gamma)\right]\\
 &=-\Tr(WA)+\Tr\!\left[(\id-W^\Gamma)Y\right].
 \label{sm:eq:Lg}
\end{align}
The infimum over \(Y\succeq0\) is finite exactly when \(\id-W^\Gamma\succeq0\).  The generalized dual is therefore
\begin{equation*}
 \max_{W\in\Herm(4)}
 \left\{-\Tr(WA):W\succeq0,\ \id-W^\Gamma\succeq0\right\}.
 \tag{$D_{\mathrm g}$}\label{sm:eq:Dg}
\end{equation*}

For \((P_{\mathrm s})\), keep \(K\succeq0\) as the variable domain and use positive semidefinite multipliers \(Q\) and \(W\) for \(K^\Gamma\succeq0\) and \(A+K\succeq0\), respectively.  Then
\begin{align}
 L_{\mathrm s}(K,Q,W)
 &=\Tr K-\Tr(QK^\Gamma)-\Tr[W(A+K)]\\
 &=-\Tr(WA)+\Tr[(\id-W-Q^\Gamma)K].
 \label{sm:eq:Ls}
\end{align}
The infimum over \(K\succeq0\) is finite exactly when \(\id-W-Q^\Gamma\succeq0\).  Naming this positive semidefinite slack \(P\) gives
\begin{equation*}
 \max\left\{-\Tr(WA):
 W,P,Q\succeq0,\ \id-W=P+Q^\Gamma\right\}.
 \tag{$D_{\mathrm s}$}\label{sm:eq:Ds}
\end{equation*}
Thus \(\id-W\) belongs to the decomposable cone and is consequently block-positive.  It need not be an entanglement witness, because it may itself be positive semidefinite.

Both primal programs satisfy Slater's condition.  Namely, for
\begin{equation}
 t>\max\{0,-\lambda_{\min}(A)\},
\end{equation}
the choice \(Y=t\id\) is strictly feasible for \((P_{\mathrm g})\), while \(K=t\id\) is strictly feasible for \((P_{\mathrm s})\).  The objectives are bounded below by zero and have finite feasible values.  Finite-dimensional semidefinite strong duality therefore gives the following identities, where \(\operatorname{val}(P)\) denotes the optimal objective value of an optimization problem \(P\):
\begin{equation}
 \operatorname{val}(P_{\mathrm g})=\operatorname{val}(D_{\mathrm g}),
 \qquad
 \operatorname{val}(P_{\mathrm s})=\operatorname{val}(D_{\mathrm s}),
 \label{sm:eq:strong-duality}
\end{equation}
and both dual optima are attained \cite{RamanaEtAl1997}.  The two dual constraints are not identical for a general high-rank \(W\).  The bridge between them will be made only after rank reduction.

\section{Auxiliary rank-reduction lemmas}\label{sm:sec:equality-proof}
\subsection{Basic two-qubit facts}
The following four elementary facts fix all normalizations and conjugations; their proofs are included for completeness.

\begin{lemma}[Block positivity of a partial transpose]
\label{sm:lem:A-block-positive}
If \(\rho\succeq0\) and \(A=\rho^\Gamma\), then \(A\) is block-positive. More explicitly, for \(\ket p=\ket e\otimes\ket f\),
\begin{equation}
 \bra pA\ket p
 =\bra{e\otimes\overline f}\rho\ket{e\otimes\overline f}\ge0,
 \label{sm:eq:A-product}
\end{equation}
where \(\overline f\) is entrywise complex conjugation in the basis used in \eqref{sm:eq:partial-transpose}.
\end{lemma}

\begin{proof}
The partial transpose of \(\proj p\) is \(\proj{e\otimes\overline f}\).  Therefore
\begin{align*}
 \bra p\rho^\Gamma\ket p
 &=\Tr(\proj p\,\rho^\Gamma)
 =\Tr((\proj p)^\Gamma\rho)\\
 &=\bra{e\otimes\overline f}\rho
   \ket{e\otimes\overline f}\ge0.
\end{align*}
\end{proof}

\begin{lemma}[Two-qubit product-vector lemma]
\label{sm:lem:product-subspace}
Every subspace \(V\subset\mathbb C^2\otimes\mathbb C^2\) with \(\dim V\ge2\) contains a nonzero product vector.
\end{lemma}

\begin{proof}
It is enough to consider a two-dimensional subspace \(V_2=\operatorname{span}\{\ket{v_1},\ket{v_2}\}\).  Identify a vector \(\sum_{i,j}v_{ij}\ket i\ket j\) with the matrix \(\widehat v=(v_{ij})\in\mathbb C^{2\times2}\).  The vector is a product vector exactly when \(\widehat v\) has rank at most one, or equivalently \(\det\widehat v=0\).  Consider the homogeneous quadratic polynomial
\[
 q(\alpha,\beta)=
 \det(\alpha\widehat v_1+\beta\widehat v_2).
\]
If \(q\) is identically zero, every nonzero vector in \(V_2\) is a product vector.  Otherwise, the fundamental theorem of algebra gives a zero \([\alpha:\beta]\in\mathbb{CP}^1\), where \(\mathbb{CP}^1\) is the complex projective line and \([\alpha:\beta]\) denotes a nonzero pair up to a common nonzero scalar multiple.  Then \(\alpha\ket{v_1}+\beta\ket{v_2}\) is the required nonzero product vector. This lemma appeared explicitly in two-qubit form in Ref.~\cite{NiuGriffiths1999}; it also follows from the fact that a completely entangled subspace of \(\mathbb C^2\otimes\mathbb C^2\) has dimension at most one \cite{Parthasarathy2004}.
\end{proof}

\begin{lemma}[Rank-one subtraction]
\label{sm:lem:rank-subtraction}
Let \(W\succeq0\), and let the normalized vector \(\ket p\) belong to \(\ran W\).  Write the spectral decomposition on the support of \(W\) as
\begin{equation}
 W=\sum_{i=1}^{k}\lambda_i\proj{w_i},
 \qquad \lambda_i>0.
\end{equation}
The Moore--Penrose pseudoinverse of \(W\) is
\begin{equation}
 W^+:=\sum_{i=1}^{k}\lambda_i^{-1}\proj{w_i};
 \label{sm:eq:pseudoinverse-definition}
\end{equation}
it acts as \(W^{-1}\) on \(\ran W\) and as zero on \(\ker W\).  Then
\begin{equation}
 t_*:=\left(\bra pW^+\ket p\right)^{-1}>0
 \label{sm:eq:tstar}
\end{equation}
satisfies
\begin{equation}
 W-t_*\proj p\succeq0,
 \qquad
 \operatorname{rank}(W-t_*\proj p)<\operatorname{rank}W.
 \label{sm:eq:rank-subtraction}
\end{equation}
\end{lemma}

\begin{proof}
Because \(p\in\ran W\), the scalar \(c=\bra pW^+\ket p\) is strictly positive.  For any \(\ket x\), Cauchy--Schwarz on the support of \(W\) gives
\[
 |\braket px|^2
 =|\langle (W^+)^{1/2}p,W^{1/2}x\rangle|^2
 \le c\,\bra xW\ket x.
\]
Thus \(W-c^{-1}\proj p\succeq0\).  If \(\ket z=W^+\ket p\), then
\[
 (W-c^{-1}\proj p)\ket z
 =\ket p-c^{-1}\ket p\bra pW^+\ket p=0.
\]
Because \(p\in\ran W\), the defining property of the pseudoinverse gives \(WW^+\ket p=\ket p\).  Hence \(W\ket z=\ket p\ne0\), which proves both that \(z\ne0\) and that \(z\in\ran W\), the support of \(W\).  The updated operator therefore has a new nonzero kernel vector inside the former support of \(W\), so its rank is strictly smaller than that of \(W\).
\end{proof}

\begin{lemma}[Largest product overlap and partial-transpose spectrum]
\label{sm:lem:rank-one-spectrum}
Let \(\ket b\) be a normalized two-qubit vector with Schmidt decomposition
\begin{equation}
 \begin{aligned}
 \ket b&=\alpha_1\ket{u_1}\ket{f_1}
       +\alpha_2\ket{u_2}\ket{f_2},\\
 \alpha_1&\ge\alpha_2\ge0,
 \qquad \alpha_1^2+\alpha_2^2=1.
 \end{aligned}
 \label{sm:eq:schmidt}
\end{equation}
Then
\begin{align}
 \max_{\substack{\ket p\text{ normalized}\\
                   \ket p\text{ product}}}
 |\braket pb|^2&=\alpha_1^2,
 \label{sm:eq:max-product-overlap}\\
 \operatorname{spec}((\proj b)^\Gamma)
 &=\{\alpha_1^2,\alpha_2^2,
     +\alpha_1\alpha_2,-\alpha_1\alpha_2\}.
 \label{sm:eq:rank-one-pt-spectrum}
\end{align}
In particular, \(\lambda_{\max}((\proj b)^\Gamma)=\alpha_1^2\).
\end{lemma}

\begin{proof}
Under local unitary changes \(U\otimes V\), the partial transpose is conjugated by \(U\otimes\overline V\).  We may therefore work in the Schmidt basis, where \(\ket b=\alpha_1\ket{00}+\alpha_2\ket{11}\).  The largest overlap with a product vector is the largest singular value of the coefficient matrix \(\operatorname{diag}(\alpha_1,\alpha_2)\), proving \eqref{sm:eq:max-product-overlap}.  This maximum is also called the product numerical radius of the rank-one projector \(\proj b\) \cite{PuchalaEtAl2011}.  Directly,
\begin{align*}
 (\proj b)^\Gamma
 &=\alpha_1^2\proj{00}+\alpha_2^2\proj{11}\\
 &\quad+\alpha_1\alpha_2
   (\ket{01}\!\bra{10}+\ket{10}\!\bra{01}).
\end{align*}
The eigenvectors \(\ket{00}\), \(\ket{11}\), and \((\ket{01}\pm\ket{10})/\sqrt2\) have the four eigenvalues in \eqref{sm:eq:rank-one-pt-spectrum}; see also the standard pure-state partial-transpose calculation in Ref.~\cite{VidalWerner2002}.
\end{proof}

\subsection{Rank-one reduction of the standard dual}
\begin{proposition}
\label{sm:prop:rank-one-standard}
If \(\rho\) is an entangled two-qubit state, the standard dual \((D_{\mathrm s})\) has an optimal triple \((W,P,Q)\) with \(\operatorname{rank}W=1\).
\end{proposition}

\begin{proof}
By strong duality and attainment, choose an optimal triple \((W,P,Q)\) in \eqref{sm:eq:Ds}.  Suppose first that \(\operatorname{rank}W\ge2\).  By Lemma~\ref{sm:lem:product-subspace}, the range of \(W\) contains a normalized product vector \(\ket p\).  Apply Lemma~\ref{sm:lem:rank-subtraction} and set
\begin{equation}
 W'=W-t_*\proj p,\qquad
 P'=P+t_*\proj p.
 \label{sm:eq:rank-update}
\end{equation}
Then \(W',P',Q\succeq0\), the rank of \(W'\) is smaller, and
\begin{equation}
 \id-W'=P'+Q^\Gamma,
\end{equation}
so the new triple remains standard-dual feasible.  Moreover,
\begin{align}
 -\Tr(W'A)
 &=-\Tr(WA)+t_*\bra pA\ket p\\
 &\ge-\Tr(WA),
 \label{sm:eq:value-nondecrease}
\end{align}
where Lemma~\ref{sm:lem:A-block-positive} gives the inequality.  Since the original triple was optimal and the new triple is feasible, strict improvement is impossible; hence equality holds in \eqref{sm:eq:value-nondecrease}, and the new triple is also optimal.

Each repetition reduces the rank by at least one.  After finitely many steps, there is an optimum with rank at most one.  It remains to exclude rank zero.  An attained zero standard-primal value would require \(K=0\), hence \(A\succeq0\), contradicting entanglement.  Strong duality therefore gives a strictly positive dual optimum, so the final \(W\) cannot be zero and must have rank one.
\end{proof}

\begin{proposition}
\label{sm:prop:rank-one-bridge}
If a standard-dual feasible triple \((W,P,Q)\) has \(\operatorname{rank}W=1\), then \(W\) is feasible for the generalized dual \((D_{\mathrm g})\).
\end{proposition}

\begin{proof}
Write \(W=w\proj b\), where \(w>0\) and \(\ket b\) is normalized.  Standard dual feasibility gives \(\id-W=P+Q^\Gamma\) with \(P,Q\succeq0\), so \(\id-W\) is block-positive.  For every normalized product vector \(\ket p\),
\begin{equation}
 0\le\bra p(\id-W)\ket p
 =1-w|\braket pb|^2.
\end{equation}
Maximizing over product vectors and using Lemma~\ref{sm:lem:rank-one-spectrum} gives
\begin{equation}
 w\alpha_1^2\le1.
 \label{sm:eq:w-alpha}
\end{equation}
The largest eigenvalue of \(W^\Gamma=w(\proj b)^\Gamma\) is exactly \(w\alpha_1^2\).  Thus \eqref{sm:eq:w-alpha} is equivalent to \(\id-W^\Gamma\succeq0\).  Together with \(W\succeq0\), these are precisely the generalized-dual constraints in \eqref{sm:eq:Dg}.
\end{proof}

\subsection{Equality and a common separable optimum}
\begin{theorem}
\label{sm:thm:equality}
For every two-qubit density operator \(\rho\),
\begin{equation}
 \Rs(\rho)=\Rg(\rho).
 \label{sm:eq:main}
\end{equation}
\end{theorem}

\begin{proof}
If \(\rho\) is separable, \(s=0\) is feasible in both definitions, so both robustnesses vanish.  Let \(\rho\) now be entangled, and set \(A=\rho^\Gamma\).  Proposition~\ref{sm:prop:rank-one-standard} supplies a rank-one standard-dual optimum \(W_{\mathrm s}\).  By standard strong duality,
\begin{equation}
 \Rs(\rho)=-\Tr(W_{\mathrm s}A).
\end{equation}
Proposition~\ref{sm:prop:rank-one-bridge} makes the same operator feasible for the generalized dual.  Generalized weak duality, or equivalently the maximization in \eqref{sm:eq:Dg}, then yields
\begin{equation}
 \Rs(\rho)=-\Tr(W_{\mathrm s}A)\le\Rg(\rho).
 \label{sm:eq:hard-direction}
\end{equation}
Combining \eqref{sm:eq:hard-direction} with the direct noise-set inclusion \eqref{sm:eq:easy-direction} proves \eqref{sm:eq:main}.
\end{proof}

\begin{corollary}
For every two-qubit state, the generalized-robustness problem has at least one optimal noise state that is separable.
\end{corollary}

\begin{proof}
If \(\rho\) is separable, choose any separable state \(\omega\) and set \(s=0\).  If \(\rho\) is entangled, let \(K_*\) be an attained optimum of the standard primal SDP \((P_{\mathrm s})\), and put
\begin{equation}
 s_*:=\Tr K_*=\Rs(\rho)>0,
 \qquad
 \omega_*:=\frac{K_*^\Gamma}{\Tr K_*}.
 \label{sm:eq:common-separable-noise-from-K}
\end{equation}
The constraints \(K_*\succeq0\) and \(K_*^\Gamma\succeq0\), together with the two-qubit PPT criterion, show that \(\omega_*\) is a normalized separable state.  Moreover, \((\rho+s_*\omega_*)^\Gamma=\rho^\Gamma+K_*\succeq0\), so the normalized mixture is separable.  Thus \(\omega_*\) is generalized-feasible with objective value \(s_*=\Rs(\rho)=\Rg(\rho)\), and is consequently generalized-optimal.
\end{proof}

\begin{remark}[Values versus optimizers]
The theorem does not state that the optimal noise is unique, that the two programs must return the same optimizer, or that every generalized-optimal noise is separable.  It states equality of the optimal values and, through the corollary, existence of at least one common separable optimum.
\end{remark}

\section{Variational formula and optimizer certificates}\label{sm:sec:variational-certificates}
\subsection{Exact single-vector formula}
The rank-one optimum also gives a variational formula.  For a normalized vector \(\ket b\), let \(\alpha_1(b)\) denote its largest Schmidt coefficient, with \(1/\sqrt2\le\alpha_1(b)\le1\), and put \([x]_+=\max\{x,0\}\).

\begin{theorem}[Exact single-vector formula]
\label{sm:thm:single-vector-formula}
For every two-qubit state \(\rho\),
\begin{equation}
 \Rs(\rho)=\Rg(\rho)=
 \max_{\substack{\ket b\in\mathbb C^2\otimes\mathbb C^2\\
                  \braket{b}{b}=1}}
 \frac{[-\bra b\rho^\Gamma\ket b]_+}{\alpha_1^2(b)}.
 \label{sm:eq:single-vector-formula}
\end{equation}
The maximum is attained.  If \(\rho\) is entangled and \(\ket{b_*}\) is obtained from a rank-one standard-dual optimum, then
\begin{equation}
 W_*=\frac{\proj{b_*}}{\alpha_1^2(b_*)}
 \label{sm:eq:saturated-rank-one-W}
\end{equation}
is optimal in both dual problems.
\end{theorem}

\begin{proof}
For every normalized \(\ket b\), Lemma~\ref{sm:lem:rank-one-spectrum} gives
\begin{equation}
 \lambda_{\max}\!\left[
 \left(\frac{\proj b}{\alpha_1^2(b)}\right)^\Gamma
 \right]=1.
\end{equation}
Consequently, \(W_b=\proj b/\alpha_1^2(b)\) satisfies \(W_b\succeq0\) and \(\id-W_b^\Gamma\succeq0\), so it is feasible for \((D_{\mathrm g})\).  Generalized weak duality therefore gives
\begin{equation}
 \Rg(\rho)\ge
 \frac{-\bra b\rho^\Gamma\ket b}{\alpha_1^2(b)}
 \label{sm:eq:single-vector-weak}
\end{equation}
for every \(b\); including the feasible dual point \(W=0\) replaces the right-hand side by its positive part.

If \(\rho\) is separable, \(\rho^\Gamma\succeq0\), so the maximum in \eqref{sm:eq:single-vector-formula} is zero.  Suppose that \(\rho\) is entangled.  Proposition~\ref{sm:prop:rank-one-standard} and Proposition~\ref{sm:prop:rank-one-bridge} give a generalized-dual optimum \(W=w\proj{b_*}\) with positive objective value and \(w\alpha_1^2(b_*)\le1\).  This inequality must be saturated.  Otherwise one could increase \(w\) until equality while retaining \(\id-W^\Gamma\succeq0\); because \(\bra{b_*}\rho^\Gamma\ket{b_*}<0\), the dual objective would strictly increase, contradicting optimality.  Thus \(w=1/\alpha_1^2(b_*)\), and this optimizer attains the right-hand side of \eqref{sm:eq:single-vector-formula}.  Continuity on the compact unit sphere also proves attainment independently.
\end{proof}

Equation~\eqref{sm:eq:single-vector-formula} is equivalent, after a change of normalization, to the rank-one/filter variational form used in the two-qubit teleportation SDP of Ref.~\cite{VerstraeteVerschelde2003}.  The present derivation shows why the same normalized vector also optimizes the standard robustness.

\subsection{Complementary slackness and optimal-noise structure}
The optimizer \(W_*\) controls more than the optimal value.  Let
\begin{equation}
 S:=\rho^\Gamma+Y^\Gamma,
 \qquad H:=\id-W_*^\Gamma
 \label{sm:eq:SH-slacks}
\end{equation}
denote the two generalized primal--dual slacks.

\begin{lemma}[Complementary slackness]
\label{sm:lem:complementary-slackness}
If \(Y\) is any generalized-primal optimum and \(W_*\) any generalized-dual optimum, then
\begin{equation}
 HY=0,
 \qquad W_*S=0.
 \label{sm:eq:matrix-complementarity}
\end{equation}
\end{lemma}

\begin{proof}
The primal--dual gap decomposes as
\begin{align}
 0&=\Tr Y+\Tr(W_*\rho^\Gamma)\notag\\
  &=\Tr[(\id-W_*^\Gamma)Y]
    +\Tr[W_*(\rho^\Gamma+Y^\Gamma)].
 \label{sm:eq:gap-decomposition}
\end{align}
Each trace on the last line is nonnegative because both factors are positive semidefinite.  Hence both vanish.  For positive semidefinite matrices \(M,N\), \(\Tr(MN)=0\) implies \(MN=0\): indeed, \(M^{1/2}NM^{1/2}\succeq0\) has zero trace and is therefore zero, which forces the supports of \(M\) and \(N\) to be orthogonal.  Applying this fact to the two terms gives \eqref{sm:eq:matrix-complementarity}.
\end{proof}

\begin{theorem}[Optimal-noise structure]
\label{sm:thm:optimal-noise-structure}
Let \(\rho\) be an entangled two-qubit state, write the common robustness as \(\Rob(\rho):=\Rs(\rho)=\Rg(\rho)>0\), and let \(\ket{b_*}\) be a maximizer in \eqref{sm:eq:single-vector-formula}. Choose its Schmidt decomposition as
\begin{equation}
 \ket{b_*}=\alpha_1\ket{u_1}\ket{f_1}
            +\alpha_2\ket{u_2}\ket{f_2},
 \qquad \alpha_1\ge\alpha_2\ge0.
 \label{sm:eq:optimal-b-schmidt}
\end{equation}
Here \(\overline{f_j}\) denotes entrywise complex conjugation in the basis used to define \(\Gamma\).  Since the optimal value is positive and \(\rho^\Gamma\) is block-positive, \(\ket{b_*}\) is entangled and hence \(\alpha_2>0\). Define \(W_*:=\proj{b_*}/\alpha_1^2\).  It is feasible for \((D_{\mathrm g})\), and maximality in \eqref{sm:eq:single-vector-formula} makes it generalized-dual optimal.

If \(\alpha_1>\alpha_2\), the generalized-primal optimum is unique and is
\begin{equation}
 Y_*=\Rob(\rho)\,
 \proj{u_1\otimes\overline{f_1}}.
 \label{sm:eq:unique-product-noise}
\end{equation}
Thus the unique normalized generalized-optimal noise is the pure product state \(\proj{u_1\otimes\overline{f_1}}\), and it is also admissible in the standard problem.

If \(\alpha_1=\alpha_2=1/\sqrt2\), complementary slackness gives only
\begin{equation}
 \ran Y\subseteq\mathcal K_{b_*}:=
 \ker(\id-W_*^\Gamma),
 \qquad \dim\mathcal K_{b_*}=3.
 \label{sm:eq:degenerate-noise-support}
\end{equation}
It does not, by itself, imply that the optimum is unique, rank one, or separable.
\end{theorem}

\begin{proof}
By the definition of \(W_*\) above and Lemma~\ref{sm:lem:rank-one-spectrum}, the eigenvalues of \(W_*^\Gamma\) are
\begin{equation}
 1,\quad \frac{\alpha_2^2}{\alpha_1^2},\quad
 \frac{\alpha_2}{\alpha_1},\quad
 -\frac{\alpha_2}{\alpha_1}.
 \label{sm:eq:optimal-W-spectrum}
\end{equation}
When \(\alpha_1>\alpha_2\), the eigenvalue one is simple and its normalized eigenvector is \(\ket{u_1}\ket{\overline{f_1}}\).  The first relation in \eqref{sm:eq:matrix-complementarity} therefore forces every optimal \(Y\) to be a nonnegative multiple of the projector onto this vector.  Since \(\Tr Y=\Rob(\rho)\), the multiple is fixed and \eqref{sm:eq:unique-product-noise} follows.  Hence the optimum is unique.

When \(\alpha_1=\alpha_2\), the spectrum in \eqref{sm:eq:optimal-W-spectrum} becomes \(\{1,1,1,-1\}\).  In Schmidt coordinates, the eigenvalue-one space is spanned by
\begin{equation}
 \ket{u_1\otimes\overline{f_1}},\quad
 \ket{u_2\otimes\overline{f_2}},\quad
 \frac{\ket{u_1\otimes\overline{f_2}}
       +\ket{u_2\otimes\overline{f_1}}}{\sqrt2}.
 \label{sm:eq:degenerate-kernel-basis}
\end{equation}
The same complementary-slackness relation proves only the support inclusion \eqref{sm:eq:degenerate-noise-support}.  A positive operator supported on this three-dimensional space can contain off-diagonal terms and need not be separable, so no stronger conclusion follows without additional primal information.
\end{proof}

\begin{figure}[t]
 \centering
 \includegraphics[width=0.92\linewidth]{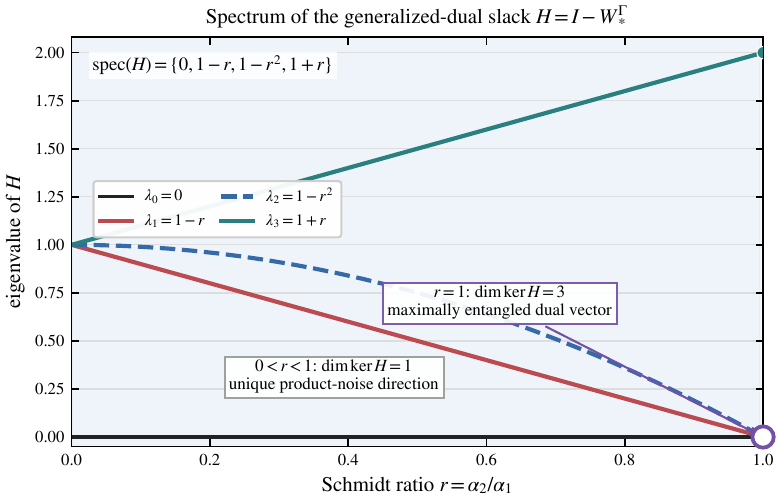}
 \caption{Spectrum of the generalized-dual slack \(H=I-W_*^\Gamma\). With Schmidt ratio \(r=\alpha_2/\alpha_1\), one has \(\operatorname{spec}(H)=\{0,1-r,1-r^2,1+r\}\). For \(0<r<1\), complementary slackness selects a unique product direction. At \(r=1\), the kernel is three-dimensional and supplies only a support restriction.}
 \label{sm:fig:dual-slack-spectrum}
\end{figure}

Figure~\ref{sm:fig:dual-slack-spectrum} displays the change of kernel dimension.

For a concrete degenerate example, take \(\rho=\proj{\Phi^+}\), where
\(\ket{\Phi^+}=(\ket{00}+\ket{11})/\sqrt2\). Its robustness is one.
Both \(Y_{\rm sep}=(\proj{01}+\proj{10})/2\) and
\(Y_{\rm ent}=\proj{\Phi^-}\), with
\(\ket{\Phi^-}=(\ket{00}-\ket{11})/\sqrt2\), have trace one and make
\(\rho+Y\) PPT. They are therefore generalized-optimal, although the
second noise is entangled. This separates the existence statement of
Theorem~\ref{sm:thm:equality} from a claim about every optimizer.

For maximally entangled maximizing vectors, the support constraint alone does not determine whether a single-product optimum exists. Proposition~\ref{sm:prop:x-minimum-products} resolves this distinction for X states; no complete preparation-size classification for arbitrary two-qubit states is claimed.

\subsection{Nonlinear kernel and inverse partial-transpose certificates}
\begin{corollary}[Nonlinear kernel certificate]
\label{sm:cor:nonlinear-kernel}
In the nondegenerate case of Theorem~\ref{sm:thm:optimal-noise-structure}, the optimal vector and common value obey
\begin{align}
 \rho^\Gamma\ket{b_*}
 &=-\Rob(\rho)\alpha_1\ket{u_1\otimes f_1},
 \label{sm:eq:nonlinear-kernel-equation}\\
 \rho^\Gamma+\Rob(\rho)\proj{u_1\otimes f_1}&\succeq0.
 \label{sm:eq:nonlinear-primal-certificate}
\end{align}
For the converse, suppose that \(\rho\) is a two-qubit density operator and that \(\ket b\) is a normalized vector with Schmidt coefficients \(\alpha_1>\alpha_2>0\).  Denote its leading Schmidt product vector by \(\ket{u_1\otimes f_1}\), and choose a real number \(r>0\).  If
\begin{align}
 \rho^\Gamma\ket b&=-r\alpha_1\ket{u_1\otimes f_1},
 \label{sm:eq:nonlinear-kernel-converse}\\
 \rho^\Gamma+r\proj{u_1\otimes f_1}&\succeq0,
 \label{sm:eq:nonlinear-primal-converse}
\end{align}
then \(r=\Rs(\rho)=\Rg(\rho)\).
\end{corollary}

\begin{proof}
Taking the partial transpose of \eqref{sm:eq:unique-product-noise} gives \(Y_*^\Gamma=\Rob(\rho)\proj{u_1\otimes f_1}\).  The second relation in \eqref{sm:eq:matrix-complementarity} implies \((\rho^\Gamma+Y_*^\Gamma)\ket{b_*}=0\), which is \eqref{sm:eq:nonlinear-kernel-equation}; primal feasibility gives \eqref{sm:eq:nonlinear-primal-certificate}.  Conversely, \eqref{sm:eq:nonlinear-primal-converse} makes \(Y=r\proj{u_1\otimes\overline{f_1}}\) generalized-primal feasible, whereas \eqref{sm:eq:nonlinear-kernel-converse} implies \(-\bra b\rho^\Gamma\ket b=r\alpha_1^2\), so \(W=\proj b/\alpha_1^2\) is generalized-dual feasible with objective \(r\).  The primal feasible point gives \(\Rg(\rho)\le r\), and the dual feasible point gives \(\Rg(\rho)\ge r\).  Thus \(\Rg(\rho)=r\), and Theorem~\ref{sm:thm:equality} gives the same value for \(\Rs(\rho)\).
\end{proof}

The converse gives a direct certificate for structured families: once \(\ket b\) and \(r\) satisfy the two displayed conditions, the leading Schmidt product vector fixes the certified optimal noise.

Every entangled two-qubit state satisfies \(\det(\rho^\Gamma)<0\), so its partial transpose is automatically nonsingular, even if \(\rho\) itself is rank deficient \cite{AugusiakDemianowiczHorodecki2008}. This gives a reduction to the product numerical range of the inverse partial transpose; the condition for a finite single-product update remains essential.

\begin{proposition}[Inverse partial-transpose reduction]
\label{sm:prop:inverse-pt-reduction}
Let \(\rho\) be an entangled two-qubit state and put \(A=\rho^\Gamma\).  Let \(\mathcal P\) be the set of normalized product vectors, and define
\begin{equation}
 q_{\min}(A^{-1}):=
 \min_{\ket p\in\mathcal P}
 \bra pA^{-1}\ket p.
 \label{sm:eq:qmin-definition}
\end{equation}
Also define the least single-product update weight by
\begin{equation}
 R_{\mathrm{prod}}(\rho):=
 \inf_{\substack{t\ge0,\ \ket p\in\mathcal P\\
                  A+t\proj p\succeq0}}t,
 \label{sm:eq:Rprod-definition}
\end{equation}
with the convention that the infimum of the empty set is \(+\infty\). For a product vector \(\ket p=\ket e\otimes\ket f\), the corresponding physical noise is \((\proj p)^\Gamma=\proj{e\otimes\overline f}\), which is again a normalized pure product state. If \(q_{\min}(A^{-1})<0\), the least feasible weight when the added noise is restricted to one pure product state is
\begin{equation}
 R_{\mathrm{prod}}(\rho)
 =-\frac{1}{q_{\min}(A^{-1})}.
 \label{sm:eq:pure-product-inverse-formula}
\end{equation}
If \(q_{\min}(A^{-1})\ge0\), no finite single-product update is feasible. If \eqref{sm:eq:single-vector-formula} has a maximizing vector with unequal Schmidt coefficients, then \(q_{\min}(A^{-1})<0\) and
\begin{equation}
 \Rob(\rho)=R_{\mathrm{prod}}(\rho)
 =-\frac{1}{q_{\min}((\rho^\Gamma)^{-1})}.
 \label{sm:eq:common-inverse-formula}
\end{equation}
\end{proposition}

\begin{proof}
An entangled two-qubit partial transpose \(A\) has inertia \((3,1,0)\) \cite{SanperaTarrachVidal1998,Rana2013,AugusiakDemianowiczHorodecki2008}. Fix a normalized product vector \(p\) and write \(q=\bra pA^{-1}\ket p\).
If \(A+t\proj p\succeq0\), the determinant identity
\begin{equation}
 \det(A+t\proj p)=\det A\,(1+tq)
 \label{sm:eq:determinant-rank-one-update}
\end{equation}
and \(\det A<0\) imply \(q<0\) and \(t\ge-1/q\).

Conversely, suppose \(q<0\), set \(z=A^{-1}p\), and use the Hermitian form
\([x,y]_A=\bra xA\ket y\). Since \([z,z]_A=q<0\), the \(A\)-orthogonal
complement of \(z\) is positive definite by Sylvester's law of inertia.
Every vector has a unique decomposition
\(x=y+cz\), where \([z,y]_A=0\) and \(c=\braket p x/q\). Therefore
\begin{equation}
 \bra x(A-q^{-1}\proj p)\ket x
 =\bra yA\ket y+|c|^2q-|cq|^2/q
 =\bra yA\ket y\ge0.
 \label{sm:eq:rank-one-inertia-update}
\end{equation}
Hence the least feasible weight for this \(p\) is \(-1/q\); every larger
weight is also feasible. This proves
\begin{equation}
 \exists\,t\ge0:\ A+t\proj p\succeq0
 \quad\Longleftrightarrow\quad q(p)<0.
 \label{sm:eq:fixed-product-equivalence}
\end{equation}
The product-vector set is compact, and \(-1/q\) is increasing on \(q<0\).
Minimization thus gives Eq.~\eqref{sm:eq:pure-product-inverse-formula};
if \(q_{\min}\ge0\), no update exists.

Finally, a maximizing vector with unequal Schmidt coefficients supplies
the optimal product noise in Theorem~\ref{sm:thm:optimal-noise-structure}.
Its partial transpose is a feasible update of weight \(R\), so
\(q_{\min}<0\) and \(R_{\rm prod}\le R\). The reverse inequality follows
because product noises form a subset of the generalized noise set.
\end{proof}

The quantity \(q_{\min}\) is the lower endpoint of the product numerical range of \((\rho^\Gamma)^{-1}\) \cite{PuchalaEtAl2011}.  Thus, in the stated nondegenerate case, evaluating the robustness reduces to a product-vector minimization for the inverse partial transpose.  When only maximally entangled maximizing vectors are available, equality \(R=R_{\rm prod}\) does not follow; the X-state classification in Proposition~\ref{sm:prop:x-minimum-products} gives explicit cases with and without a single-product optimum.

\section{Exact X-state formula and optimal noises}\label{sm:sec:x-values}

Local phase unitaries put every two-qubit X state into the real canonical form
\begin{equation}
 \rho_X=
 \begin{pmatrix}
 a&0&0&w\\
 0&b&z&0\\
 0&z&c&0\\
 w&0&0&d
 \end{pmatrix},
 \quad
 \begin{gathered}
 a,b,c,d,w,z\ge0,\\
 a+b+c+d=1,
 \end{gathered}
 \label{sm:eq:x-state}
\end{equation}
with physicality conditions
\begin{equation}
 w^2\le ad,
 \qquad z^2\le bc.
 \label{sm:eq:x-physicality}
\end{equation}
After a basis permutation, its partial transpose is
\begin{equation}
 \rho_X^\Gamma\cong
 \begin{pmatrix}a&z\\z&d\end{pmatrix}
 \oplus
 \begin{pmatrix}b&w\\w&c\end{pmatrix}.
 \label{sm:eq:x-pt-blocks}
\end{equation}
Consequently, the state is entangled precisely when either \(z^2>ad\) or \(w^2>bc\); the two violations cannot occur simultaneously because of Eq.~\eqref{sm:eq:x-physicality}.

\begin{theorem}[Two-qubit X-state robustness]
\label{sm:thm:x-state}
If \(z^2>ad\), then
\begin{equation}
 \Rg(\rho_X)=\Rs(\rho_X)=
 \begin{cases}
 \dfrac{z^2-ad}{\max\{a,d\}},
 &\sqrt{ad}<z\le\max\{a,d\},\\[1.2ex]
 2z-a-d,
 &z>\max\{a,d\}.
 \end{cases}
 \label{sm:eq:x-formula-z}
\end{equation}
If \(w^2>bc\), the formula follows from the simultaneous replacement
\begin{equation}
 (a,d,z)\longmapsto(b,c,w).
 \label{sm:eq:x-branch-swap}
\end{equation}
If neither strict inequality holds, both robustnesses vanish.
\end{theorem}

\begin{proof}
It is enough to prove the \(z^2>ad\) case. The local Pauli operation \(\sigma_x\otimes\sigma_x\), where \(\sigma_x=\ket0\!\bra1+\ket1\!\bra0\), interchanges \(a\) and \(d\) without changing the robustness. We may assume \(d\ge a\). In particular, \(z>\sqrt{ad}\ge a\), and the two branches are \(a<z\le d\) and \(z>d\).

For \(0\le\mu\le1\), define the normalized test vector
\begin{equation}
 \ket{\beta(\mu)}
 =\frac{\ket{00}-\mu\ket{11}}{\sqrt{1+\mu^2}}.
 \label{sm:eq:x-beta-mu}
\end{equation}
Its largest Schmidt coefficient squared is \(1/(1+\mu^2)\). From the even block in Eq.~\eqref{sm:eq:x-pt-blocks},
\begin{equation}
 \bra{\beta(\mu)}\rho_X^\Gamma\ket{\beta(\mu)}
 =\frac{a-2z\mu+d\mu^2}{1+\mu^2}.
\end{equation}
Equation~\eqref{sm:eq:single-vector-formula} therefore yields
\begin{equation}
 \Rg(\rho_X)\ge L(\mu):=2z\mu-a-d\mu^2.
 \label{sm:eq:x-L-mu}
\end{equation}
If \(d=0\), the convention \(d\ge a\ge0\) gives \(a=0\), while \(z^2>ad\) gives \(z>0\). This case belongs to the branch \(z>d\), and \(L(\mu)=2z\mu\) is maximized at \(\mu_*=1\). If \(d>0\), the concave quadratic \(L\) is maximized on \([0,1]\) at \(\mu_* =\min\{z/d,1\}\). Hence
\begin{equation}
 \Rg(\rho_X)\ge
 \begin{cases}
 (z^2-ad)/d,&z\le d,\\
 2z-a-d,&z>d.
 \end{cases}
 \label{sm:eq:x-lower-branches}
\end{equation}

Suppose first that \(z\le d\). This branch has \(d>0\). Set
\begin{equation}
 s_1=\frac{z^2-ad}{d}>0,
 \qquad \omega_1=\proj{00}.
 \label{sm:eq:x-noise-first}
\end{equation}
The noise \(\omega_1\) is a pure product state, and \(\rho_X+s_1\omega_1\succeq0\). The odd block of its partial transpose is unchanged and positive because the two PPT violations cannot occur simultaneously. The even block is
\begin{equation}
 \begin{pmatrix}a+s_1&z\\z&d\end{pmatrix},
 \qquad d(a+s_1)-z^2=0,
 \label{sm:eq:x-first-pt-block}
\end{equation}
and is positive semidefinite. The normalized mixture is therefore a two-qubit PPT state and hence separable. Thus \(\Rs(\rho_X)\le s_1\). Together with \(\Rg\le\Rs\) and the first line of Eq.~\eqref{sm:eq:x-lower-branches}, this proves the first branch.

Suppose next that \(z>d\). Define
\begin{equation}
 \begin{aligned}
 s_2&=2z-a-d,
 &t&=\frac{z-a}{s_2},\\
 \omega_2&=t\proj{00}+(1-t)\proj{11}.&&
 \end{aligned}
 \label{sm:eq:x-noise-second}
\end{equation}
Because \(z>d\ge a\), one has \(s_2>0\), \(0<t<1\), and
\begin{equation}
 1-t=\frac{z-d}{s_2}.
\end{equation}
Thus \(\omega_2\) is a convex combination of pure product states. Moreover,
\begin{equation}
 a+s_2t=d+s_2(1-t)=z,
\end{equation}
so the even block of \((\rho_X+s_2\omega_2)^\Gamma\) is
\begin{equation}
 \begin{pmatrix}z&z\\z&z\end{pmatrix}\succeq0.
\end{equation}
Its odd block is again unchanged and positive, and the untransposed sum is positive. Hence the normalized mixture is separable and \(\Rs(\rho_X)\le s_2\). The second line of Eq.~\eqref{sm:eq:x-lower-branches} and \(\Rg\le\Rs\) prove the second branch.

Restoring the symmetry between \(a\) and \(d\) replaces \(d\) by \(\max\{a,d\}\). Exchanging the even and odd parity sectors proves Eq.~\eqref{sm:eq:x-branch-swap}. If neither determinant is negative, Eq.~\eqref{sm:eq:x-pt-blocks} is positive semidefinite, so \(\rho_X\) is PPT and separable.

The two branches also illustrate Theorem~\ref{sm:thm:optimal-noise-structure}. Under the convention \(d\ge a\), the first branch has optimal test vector \(\ket{\beta(z/d)}\). For \(z<d\), its Schmidt coefficients are nondegenerate, so \(\omega_1=\proj{00}\) is the unique generalized-optimal noise state. In the second branch, \(\ket{\beta(1)}=(\ket{00}-\ket{11})/\sqrt2\) is maximally entangled and the dual slack has a three-dimensional kernel. The explicit two-product mixture \(\omega_2\) proves optimality, but the dual support condition alone does not prove uniqueness. The boundary case \(z=d\) also lies in this degenerate dual regime, although the one-product construction remains valid.
\end{proof}

\section{Comparison with previous two-qubit robustness formulas}
\label{sm:sec:prior-formula-comparison}

This section compares the Letter with earlier robustness and teleportation calculations. We use the noise-to-signal convention of Eqs.~\eqref{sm:eq:def-rs} and \eqref{sm:eq:def-rg}: the robustness is the coefficient \(s\), not \(1+s\). Wootters' concurrence is \(C=\max\{0,\lambda_1-\lambda_2-\lambda_3-\lambda_4\}\), where the \(\lambda_i\) are the decreasing square roots of the eigenvalues of \(\rho\widetilde\rho\), and \(\widetilde\rho=(\sigma_y\otimes\sigma_y)\overline\rho(\sigma_y\otimes\sigma_y)\) \cite{Wootters1998}. The finite-versus-infinite statements below are unaffected by the shifted convention used in Refs.~\cite{RegulaLamiFerrariTakagi2021,LamiRegulaTakagiFerrari2021}.

\subsection{The Wootters-decomposition formula and its domain}

Section~4.1 of Ref.~\cite{JafarizadehMirzaeeRezaee2005}, pp.~515--521, starts from a Wootters decomposition
\begin{equation}
 \rho=\sum_{i=1}^4\proj{x_i},
 \qquad
 \braket{x_i}{\widetilde{x_j}}=\lambda_i\delta_{ij},
 \qquad
 \lambda_1\geq\lambda_2\geq\lambda_3\geq\lambda_4,
 \label{sm:eq:jmr-wootters-decomposition}
\end{equation}
and defines
\begin{equation}
 \ket{x_i'}=\frac{\ket{x_i}}{\sqrt{\lambda_i}},
 \qquad
 K_i=\braket{x_i'}{x_i'},
 \qquad
 P_i=\lambda_iK_i>0.
 \label{sm:eq:jmr-normalized-vectors}
\end{equation}
These are the definitions in Eqs.~(7)--(13) of that reference. For the entangled sector \(C=\lambda_1-\lambda_2-\lambda_3-\lambda_4>0\), Eq.~(39), p.~520, reports
\begin{equation}
 R_{\rm JMR}(\rho)
 =\frac{C}{2}
 \min\{K_2+K_3,K_2+K_4,K_3+K_4\}.
 \label{sm:eq:jmr-reported-formula}
\end{equation}
Equation~(6) defines absolute robustness by minimizing over all separable noise states. The calculation first restricts noise and target to the cone generated by the four fixed Wootters projectors. After Eq.~(40), the authors explicitly acknowledge this restriction and attempt to extend the minimum to off-diagonal separable states through Eqs.~(41)--(48), p.~521. Thus the issue is the claimed full-space optimality, not the existence of a restricted-family construction. All external equation and section numbers in this subsection refer to Ref.~\cite{JafarizadehMirzaeeRezaee2005}.

Equations~\eqref{sm:eq:jmr-normalized-vectors} require nonzero \(\lambda_i\), while a degenerate Wootters spectrum can make the frame nonunique. Neither issue is needed for the following test: our counterexample is full rank and has four positive, pairwise distinct Wootters numbers. It therefore tests the generic domain directly.

For a strictly positive real X state in Eq.~\eqref{sm:eq:x-state}, with strict physicality inequalities, the even- and odd-parity Wootters pairs can be chosen so that
\begin{align}
 \lambda_{e,\pm}&=\sqrt{ad}\pm w,
 &K_{e,\pm}&=\kappa_e:=\frac{a+d}{2\sqrt{ad}},\label{sm:eq:x-wootters-even}\\
 \lambda_{o,\pm}&=\sqrt{bc}\pm z,
 &K_{o,\pm}&=\kappa_o:=\frac{b+c}{2\sqrt{bc}}.\label{sm:eq:x-wootters-odd}
\end{align}
For example, with the phase chosen to make both spin-flip overlaps nonnegative, the even Wootters vectors may be written
\begin{align}
 \ket{x_{e,+}}
 &=i\sqrt{\frac{\sqrt{ad}+w}{2\sqrt{ad}}}
 \left(\sqrt a\ket{00}+\sqrt d\ket{11}\right),\\
 \ket{x_{e,-}}
 &=\sqrt{\frac{\sqrt{ad}-w}{2\sqrt{ad}}}
 \left(\sqrt a\ket{00}-\sqrt d\ket{11}\right).
\end{align}
Division by \(\sqrt{\lambda_{e,\pm}}\) gives squared norm \(\kappa_e\), and the odd pair is analogous. If \(w>\sqrt{bc}\), then \(C=2(w-\sqrt{bc})\), the dominant vector is the \((e,+)\) vector, and direct substitution into Eq.~\eqref{sm:eq:jmr-reported-formula} gives
\begin{equation}
 R_{\rm JMR}^{(w)}
 =C\min\left\{\kappa_o,\frac{\kappa_e+\kappa_o}{2}\right\}.
 \label{sm:eq:jmr-x-specialization-w}
\end{equation}
If \(z>\sqrt{ad}\), then \(C=2(z-\sqrt{ad})\) and instead
\begin{equation}
 R_{\rm JMR}^{(z)}
 =C\min\left\{\kappa_e,\frac{\kappa_e+\kappa_o}{2}\right\}.
 \label{sm:eq:jmr-x-specialization-z}
\end{equation}
These expressions are not the piecewise values in Theorem~\ref{sm:thm:x-state}; in particular, Eq.~\eqref{sm:eq:jmr-x-specialization-w} can depend on the spectator block \((a,d)\), whereas the exact \(w\)-branch value depends only on \((b,c,w)\).
The strict assumptions above are needed only to test the domain of Eq.~\eqref{sm:eq:jmr-reported-formula}. Theorem~\ref{sm:thm:x-state} independently covers ranks 1--4, zero diagonal entries, the PPT boundaries \(w=\sqrt{bc}\) and \(z=\sqrt{ad}\), and the internal branch boundaries \(w=\max\{b,c\}\) and \(z=\max\{a,d\}\), where its two displayed expressions agree.

An exact example removes any possible ambiguity due to zero Wootters numbers or degeneracy. Let
\begin{equation}
 \rho_{\rm ex}=
 \begin{pmatrix}
  3/8&0&0&1/4\\
  0&1/5&1/100&0\\
  0&1/100&1/20&0\\
  1/4&0&0&3/8
 \end{pmatrix}.
 \label{sm:eq:jmr-counterexample-state}
\end{equation}
It is full rank, with
\begin{equation}
 \operatorname{spec}(\rho_{\rm ex})
 =\left\{\frac58,\frac18,
 \frac18+\frac{\sqrt{229}}{200},
 \frac18-\frac{\sqrt{229}}{200}\right\}.
\end{equation}
Its four strictly positive, pairwise distinct Wootters numbers and the corresponding norms are
\begin{equation}
 (\lambda_1,\lambda_2,\lambda_3,\lambda_4)
 =\left(\frac58,\frac18,\frac{11}{100},\frac9{100}\right),
 \qquad
 (K_1,K_2,K_3,K_4)=\left(1,1,\frac54,\frac54\right),
 \label{sm:eq:jmr-counterexample-data}
\end{equation}
so \(C=3/10\) and Eq.~\eqref{sm:eq:jmr-reported-formula} predicts
\begin{equation}
 R_{\rm JMR}(\rho_{\rm ex})
 =\frac{3}{20}\frac94=\frac{27}{80}.
 \label{sm:eq:jmr-counterexample-prediction}
\end{equation}
The exact value is instead \(1/4\), as witnessed independently from both sides. The separable unnormalized noise
\begin{equation}
 Y_{\rm ex}=\frac1{20}\proj{01}+\frac15\proj{10},
 \qquad \Tr Y_{\rm ex}=\frac14,
 \label{sm:eq:jmr-counterexample-primal}
\end{equation}
gives
\begin{equation}
 (\rho_{\rm ex}+Y_{\rm ex})^\Gamma\cong
 \begin{pmatrix}3/8&1/100\\1/100&3/8\end{pmatrix}
 \oplus
 \begin{pmatrix}1/4&1/4\\1/4&1/4\end{pmatrix}
 \succeq0.
 \label{sm:eq:jmr-counterexample-ppt}
\end{equation}
Thus it is standard-primal feasible. Conversely, with
\begin{equation}
 \ket{\beta}=\frac{\ket{01}-\ket{10}}{\sqrt2},
 \qquad W_{\rm ex}=2\proj{\beta},
\end{equation}
one has \(W_{\rm ex}\succeq0\),
\begin{equation}
 \operatorname{spec}(\id-W_{\rm ex}^\Gamma)=\{2,0,0,0\},
 \qquad
 -\Tr(W_{\rm ex}\rho_{\rm ex}^\Gamma)=\frac14.
 \label{sm:eq:jmr-counterexample-dual}
\end{equation}
Hence \(W_{\rm ex}\) is generalized-dual feasible and
\begin{equation}
 \frac14\leq\Rg(\rho_{\rm ex})\leq\Rs(\rho_{\rm ex})\leq\frac14.
\end{equation}
This proves the exact value without invoking Theorem~\ref{sm:thm:equality} and disproves Eq.~\eqref{sm:eq:jmr-reported-formula} as a universal two-qubit optimum. Equation~(40) selects a minimizing-pair noise; below we exhibit one such feasible candidate and a strictly improving path. Equation~\eqref{sm:eq:jmr-counterexample-primal} instead supplies an optimum of weight \(1/4\). The Bell-decomposable result has an independent derivation \cite{AkhtarshenasJafarizadeh2003} and is also treated separately in Sec.~4.2 of Ref.~\cite{JafarizadehMirzaeeRezaee2005}. For that family \(K_i=1\) and \(R=C\), consistently with the present theorem. The counterexample does not invalidate this special case or establish that all other results of Ref.~\cite{JafarizadehMirzaeeRezaee2005} are incorrect.

The fixed Wootters tetrahedron is not the full separable set. For \(\rho_{\rm ex}\), each odd-parity Wootters projector has diagonal ratio \(4:1\), whereas \(Y_{\rm ex}\) has ratio \(1:4\). Even projectors cannot help in a positive combination because \(Y_{\rm ex}\) has zero even block. Thus the certified optimum lies outside the restricted family.

The proposed full-space extension in Sec.~4.1 of Ref.~\cite{JafarizadehMirzaeeRezaee2005} uses nonorthogonal-frame and dual-frame coefficients in Eqs.~(41)--(43). The pseudomixture identities (44)--(47) include cancellation of the off-diagonal coefficients; Eq.~(48) is then used to infer that such coefficients do not affect the optimum. Cancellation alone does not justify that inference: the diagonal coefficients appearing in the ratio can change when the admissible separable states change. Nor is deletion of off-diagonal terms in a nonorthogonal frame shown to preserve both normalization and separability. The following feasible descent directly tests the resulting optimality claim, without assuming such a deletion map.

An explicit feasible descent rules out even local optimality of the old candidate. Choose its minimizing Wootters pair \(2,3\), and let
\begin{equation}
 Y_0=\frac3{20}\left(\proj{\Phi^-}+\proj{\chi}\right),\qquad
 \ket{\Phi^-}=\frac{\ket{00}-\ket{11}}{\sqrt2},\qquad
 \ket{\chi}=\frac{2\ket{01}+\ket{10}}2 .
 \label{sm:eq:jmr-old-noise-vectors}
\end{equation}
Here \(\chi\) is intentionally unnormalized, with squared norm \(5/4\), as required by the Wootters-frame convention. Explicitly,
\begin{equation}
 Y_0=\frac1{80}
 \begin{pmatrix}
 6&0&0&-6\\0&12&6&0\\0&6&3&0\\-6&0&0&6
 \end{pmatrix},\qquad \Tr Y_0=\frac{27}{80}.
 \label{sm:eq:jmr-old-noise-matrix}
\end{equation}
Both \(Y_0\) and \(Y_0^\Gamma\) are positive semidefinite. The target is also PPT:
\begin{equation}
 (\rho_{\rm ex}+Y_0)^\Gamma\cong
 \begin{pmatrix}9/20&17/200\\17/200&9/20\end{pmatrix}
 \oplus
 \begin{pmatrix}7/20&7/40\\7/40&7/80\end{pmatrix}
 \succeq0.
 \label{sm:eq:jmr-old-noise-target}
\end{equation}
The entire path
\begin{equation}
 Y_\tau=(1-\tau)Y_0+\tau Y_{\rm ex},\qquad
 \Tr Y_\tau=\frac{27}{80}-\frac{7\tau}{80},
 \qquad 0\le\tau\le1,
 \label{sm:eq:jmr-feasible-descent}
\end{equation}
is standard-primal feasible by convexity of the separable cone. For every \(\tau>0\), however small, its cost is strictly lower. The normalized noise \(Y_\tau/\Tr Y_\tau\) and normalized separable target vary continuously at \(\tau=0\). Thus the restricted optimum is not a full-space local minimum; the proposed perturbative extension cannot exclude these nearby feasible noises.

The comparison of scope and conclusions is summarized below.
\begin{table}[ht]
\centering
\small
\renewcommand{\arraystretch}{1.18}
\begin{tabular}{p{0.23\linewidth}p{0.29\linewidth}p{0.40\linewidth}}
\hline\hline
Situation & Ref.~\cite{JafarizadehMirzaeeRezaee2005}, Sec.~4.1 & Present treatment \\
\hline
Full-rank generic case & Eq.~\eqref{sm:eq:jmr-reported-formula} is claimed & Exact state \eqref{sm:eq:jmr-counterexample-state} contradicts it \\
Rank deficient or \(\lambda_i=0\) & Displayed rescaling requires a limiting treatment & Direct conic proof includes all ranks \\
Degenerate \(\lambda_i\) & Frame may be nonunique & Counterexample has no degeneracy \\
Restricted versus full space & Eqs.~(41)--(48) propose an extension & A strictly improving feasible path disproves full-space optimality \\
Bell decomposable & \(K_i=1\), hence \(R=C\) & Recovered as a consistency check \\
\hline\hline
\end{tabular}
\caption{Domain and conclusion of the earlier Wootters-decomposition calculation.}
\label{sm:tab:jmr-scope-audit}
\end{table}

\subsection{From the teleportation SDP to standard robustness}

Verstraete and Verschelde optimize the deterministic teleportation singlet fraction as \cite{VerstraeteVerschelde2003}
\begin{equation}
 F^*(\rho)=\max_X\left\{\frac12-\Tr(X\rho^\Gamma):
 0\preceq X\preceq\id,
 -\frac{\id}{2}\preceq X^\Gamma\preceq\frac{\id}{2}\right\}.
 \label{sm:eq:vv-teleportation-primal}
\end{equation}
Their SDP dual is
\begin{equation}
 F^*(\rho)=\min_Z\left\{\frac12+\frac12\Tr Z:
 Z\succeq0, (\rho+Z)^\Gamma\succeq0\right\}.
 \label{sm:eq:vv-teleportation-dual}
\end{equation}
Since PPT is equivalent to separability in \(2\otimes2\), identifying \(Z=Y\) in \((P_{\rm g})\) gives the exact normalization
\begin{equation}
 \Rg(\rho)=2F^*(\rho)-1.
 \label{sm:eq:vv-rg-normalization}
\end{equation}
But \(Z\) in Eq.~\eqref{sm:eq:vv-teleportation-dual} is an arbitrary positive operator. Standard robustness imposes the additional condition \(Z^\Gamma\succeq0\), equivalently that the normalized noise be separable. Nothing in Eq.~\eqref{sm:eq:vv-teleportation-dual} shows that this additional constraint is cost-free.

The rank-one reduction in Ref.~\cite{VerstraeteVerschelde2003} gives an equivalent optimum of \((D_{\rm g})\). Explicitly, \(W=2X\) maps every point in Eq.~\eqref{sm:eq:vv-teleportation-primal} to a generalized-dual feasible point. Conversely, for a rank-one generalized-dual optimum \(W=w\proj b\), the bound \(w\alpha_1^2\le1\) gives \(w\le2\) and \(w\alpha_1\alpha_2\le1\). Thus \(X=W/2\) obeys \(X\preceq I\) and \(-I/2\preceq X^\Gamma\preceq I/2\), with the same rescaled objective. This establishes equality of these generalized optimal values without asserting equality of all feasible sets.

A generalized-dual feasible operator is automatically standard-dual feasible: if \(\id-W^\Gamma\succeq0\), then
\begin{equation}
 \id-W=(\id-W^\Gamma)^\Gamma
\end{equation}
is of the decomposable form required by \((D_{\rm s})\). Thus starting from the generalized SDP yields only \(\Rg\leq\Rs\), the easy direction already implied by inclusion of the primal noise sets. The nontrivial direction requires the converse at an optimum. Our proof starts with a \emph{standard}-dual optimizer, reduces it to rank one, and only then proves \(\id-W^\Gamma\succeq0\). This establishes \(\Rs\leq\Rg\), equality of the values, and---by finite-dimensional attainment---existence of a generalized-primal optimum whose noise is separable. It does not claim that every generalized optimizer is separable.
\subsection{Separation in different dimensions}
Regula, Lami, Ferrari, and Takagi exhibited an infinite-dimensional entangled state with finite generalized robustness but infinite standard robustness \cite{RegulaLamiFerrariTakagi2021}. The companion PRA develops the Hilbert-operator construction in detail \cite{LamiRegulaTakagiFerrari2021}. In their convention, one is added to both values. Finite truncations of the construction imply separation in sufficiently large finite dimensions. The following truncation argument makes this finite-dimensional consequence explicit.

For completeness, the finite-truncation deduction uses the same weights
\(d_j=[\sqrt j\ln(j+1)]^{-1}\), \(c_N=\sum_{j=1}^Nd_j^2\), and
\((H_N)_{jk}=1/(j-k)\) for \(j\ne k\), with zero diagonal. Since
\(\|H_N\|\le\pi\), the maximally correlated matrices
\[
 \rho_{\pm,N}=\frac1{c_N}\sum_{j,k=1}^N
 d_jd_k\left(\delta_{jk}\pm\frac{i}{\pi}(H_N)_{jk}\right)
 \ket{jj}\bra{kk}
\]
are states, and their sum is separable, so \(\Rg(\rho_{\pm,N})\le1\).
For any standard feasible noise of trace \(s\), the triangle inequality
for the partial-transpose trace norm gives
\(\|\rho^\Gamma\|_1\le1+2s\); hence
\[
 \Rs(\rho_{\pm,N})\ge\mathcal N(\rho_{\pm,N})
 =\frac1{\pi c_N}\sum_{1\le j<k\le N}\frac{d_jd_k}{k-j}
 \longrightarrow\infty.
\]
The divergence follows already by restricting to \(j<k\le2j\): the
corresponding sum for each \(j\) is at least
\(1/[\sqrt2\,j\ln(2j+1)]\), whereas \(c_N\) stays bounded.
Thus some finite truncation separates the measures, without locating
the minimal dimension. This deduction uses the construction and Hilbert
norm bound of Ref.~\cite{LamiRegulaTakagiFerrari2021}.

More directly, Lami and Regula gave the two-qutrit state \cite{LamiRegula2023}
\begin{equation}
 \omega_3=\frac{P_3-\Phi_3}{2},\qquad
 P_3=\sum_{j=1}^3\proj{jj},\qquad
 \Phi_3=\frac13\sum_{j,k=1}^3\ket{jj}\bra{kk}.
 \label{sm:eq:lr-qutrit-state}
\end{equation}
Their Supplemental Note VI A, Eqs.~(S138)--(S139), gives, in the present unshifted convention,
\begin{equation}
 \Rs(\omega_3)=\frac34,\qquad
 \Rg(\omega_3)=\frac12.
 \label{sm:eq:lr-qutrit-values}
\end{equation}
For example, \(\omega_3+\Phi_3/2=P_3/2\) is their optimal generalized decomposition. They also use this state to study entanglement irreversibility.

\begin{table}[ht]
\caption{Earlier results and the dimension boundary.}
\label{sm:tab:nearest-prior-work}
\small
\begin{tabular}{p{0.19\linewidth}p{0.36\linewidth}p{0.36\linewidth}}
\hline
\raggedright Work & \raggedright Established result & \raggedright Relation to this Letter\tabularnewline
\hline
\raggedright Verstraete--Verschelde (2003) & \raggedright Generalized SDP, rank-one/filter formulation, teleportation relation & \raggedright Their rank-one reduction provides a starting point for our standard-dual argument.\tabularnewline
\raggedright Regula \emph{et al.}, Lami \emph{et al.} (2021) & \raggedright Finite generalized/infinite standard robustness & \raggedright Finite truncations give separation in sufficiently large dimensions.\tabularnewline
\raggedright Lami--Regula (2023) & \raggedright Explicit \(3\otimes3\) separation with exact values \(3/4\) and \(1/2\); entanglement irreversibility & \raggedright We establish universal equality in \(2\otimes2\) and separation already in \(2\otimes3\).\tabularnewline
\raggedright This Letter & \raggedright All-state no-advantage theorem; rational \(2\otimes3\) certificate, stable at specified full-rank perturbations & \raggedright Local isometry invariance completes the finite-dimensional boundary, also for a fixed multipartite cut.\tabularnewline
\hline
\end{tabular}
\end{table}

An example in \(3\otimes3\) leaves open whether entangled replacement can help in the smallest systems, and whether PPT exactness guarantees equality. The all-state theorem and the \(2\otimes3\) certificate answer these questions differently: two-qubit replacement never requires entanglement at the optimum, whereas qubit--qutrit replacement can benefit despite PPT remaining exact. The product-vector step explains why the equality proof is special to two qubits; the certificate establishes failure of equality itself.

These statements concern exact, single-copy robustness. Tensor powers, smoothing, and asymptotic tasks require separate analysis.

\section{Two-product optimal noises and the BSA remainder of X states}
\label{sm:sec:x-two-product-bsa}

Section~\ref{sm:sec:x-values} gives the direct primal--dual value proof. Here we independently compress arbitrary feasible noise to two product directions, prove the minimum preparation size, and compare the BSA remainder. We use its canonical X-state notation and set
\begin{equation}
 \Pi=\sigma_z\otimes\sigma_z,
 \qquad
 \mathcal T(X)=\frac12\bigl(X+\Pi X\Pi\bigr).
 \label{sm:eq:x-parity-twirl}
\end{equation}
The map \(\mathcal T\) is a mixture of local unitaries. It therefore preserves positivity, trace, and the separable cone, and every X state obeys \(\mathcal T(\rho_X)=\rho_X\).

\begin{proposition}[Two-product reduction]
\label{sm:prop:x-two-product-reduction}
Suppose that \(\rho_X\) lies in the \(w\)-entangled branch, \(\lvert w\rvert^2>bc\). For every generalized-primal feasible operator \(N\),
\begin{equation}
 N\succeq0,
 \qquad
 \rho_X+N\in\SEP_+,
 \label{sm:eq:x-general-feasible-noise}
\end{equation}
there exist \(u,v\geq0\) such that
\begin{equation}
 N'=u\proj{01}+v\proj{10},
 \quad
 \Tr N'=\Tr N,
 \quad
 \rho_X+N'\in\SEP_+.
 \label{sm:eq:x-two-product-replacement-w}
\end{equation}
Consequently, at least one generalized optimum, and hence at least one standard optimum, is supported on \(\{\ket{01},\ket{10}\}\). In the \(z\)-entangled branch, \(\lvert z\rvert^2>ad\), the corresponding support is \(\{\ket{00},\ket{11}\}\).
\end{proposition}

\begin{proof}
Apply \(\mathcal T\) to \(N\). Equation~\eqref{sm:eq:x-general-feasible-noise} remains valid and the trace is unchanged, because \(\rho_X+\mathcal T(N)=\mathcal T(\rho_X+N)\). We may therefore write
\begin{equation}
 N=
 \begin{pmatrix}
  \alpha&0&0&\xi\\
  0&\beta&\eta&0\\
  0&\eta^*&\gamma&0\\
  \xi^*&0&0&\delta
 \end{pmatrix},
 \qquad
 \alpha,\beta,\gamma,\delta\geq0.
 \label{sm:eq:x-general-twirled-noise}
\end{equation}
Positivity of the even-parity principal block gives \(\lvert\xi\rvert\leq\sqrt{\alpha\delta}\). Since \(\rho_X+N\) is separable, its partial transpose is positive semidefinite. The odd-parity block containing \(w+\xi\) therefore gives
\begin{equation}
 \lvert w+\xi\rvert\leq\sqrt{(b+\beta)(c+\gamma)}.
 \label{sm:eq:x-w-xi-bound}
\end{equation}
The triangle inequality implies
\begin{equation}
 \lvert w\rvert\leq\sqrt{(b+\beta)(c+\gamma)}+\sqrt{\alpha\delta}.
 \label{sm:eq:x-w-compression-bound}
\end{equation}
For nonnegative \(X,Y,\alpha,\delta\),
\begin{equation}
 (X+\alpha)(Y+\delta)
 -\bigl(\sqrt{XY}+\sqrt{\alpha\delta}\bigr)^2
 =\bigl(\sqrt{X\delta}-\sqrt{Y\alpha}\bigr)^2\geq0.
 \label{sm:eq:x-compression-identity}
\end{equation}
Using \(X=b+\beta\) and \(Y=c+\gamma\) yields
\begin{equation}
 (b+\alpha+\beta)(c+\gamma+\delta)\geq\lvert w\rvert^2.
 \label{sm:eq:x-compressed-ppt-condition}
\end{equation}
Set \(u=\alpha+\beta\) and \(v=\gamma+\delta\). Then \(N'=u\proj{01}+v\proj{10}\) is separable and has the same trace as \(N\). Equation~\eqref{sm:eq:x-compressed-ppt-condition} makes the odd-parity block of \((\rho_X+N')^\Gamma\) positive semidefinite. Its other block is positive because physicality and \(\lvert w\rvert^2>bc\) imply \(ad\geq\lvert w\rvert^2>bc\geq\lvert z\rvert^2\). Thus \(\rho_X+N'\succeq0\) and \((\rho_X+N')^\Gamma\succeq0\). The PPT criterion is sufficient for separability in \(2\otimes2\), proving \(\rho_X+N'\in\SEP_+\). Exchanging the even- and odd-parity sectors proves the \(z\)-branch statement.
\end{proof}

For the \(w\)-entangled branch, Proposition~\ref{sm:prop:x-two-product-reduction} reduces both robustness optimizations to
\begin{equation}
 \min_{u,v\geq0}(u+v)
 \quad\text{subject to}\quad
 (b+u)(c+v)\geq t^2,
 \qquad t:=\lvert w\rvert.
 \label{sm:eq:x-two-variable-program-w}
\end{equation}
The independent compression proof therefore recovers the value in Theorem~\ref{sm:thm:x-state}, with weights
\begin{equation}
 (u_*,v_*)=
 \begin{cases}
  (t-b,t-c),&t\geq\max\{b,c\},\\[1mm]
  (0,t^2/b-c),&b>t,\\[1mm]
  (t^2/c-b,0),&c>t.
 \end{cases}
 \label{sm:eq:x-optimal-two-product-weights-w}
\end{equation}
The constraint is active at an optimum. With \(x=b+u\) and \(y=c+v\), one minimizes \(x+y\) subject to \(xy=t^2\), \(x\geq b\), and \(y\geq c\). If \(t\geq\max\{b,c\}\), the arithmetic--geometric mean inequality is saturated by \(x=y=t\). If \(b>t\), the function \(x+t^2/x\) is increasing for \(x\geq b\), so \(x=b\); the case \(c>t\) is symmetric. The latter two cases cannot occur simultaneously because \(t^2>bc\). Hence
\begin{equation}
 \Rob(\rho_X)=u_*+v_*,
 \qquad
 \omega_*=\frac{u_*\proj{01}+v_*\proj{10}}{u_*+v_*}
 \label{sm:eq:x-optimal-normalized-noise-w}
\end{equation}
is an optimal separable noise state. The \(z\)-branch follows by
\begin{equation}
 (b,c,w;\ket{01},\ket{10})
 \longleftrightarrow
 (a,d,z;\ket{00},\ket{11}).
 \label{sm:eq:x-two-product-branch-exchange}
\end{equation}
The weights specify a representative optimum; the full-space minimum-size statement requires the following additional obstruction.

\begin{proposition}[Minimum pure-product preparation size for X states]
\label{sm:prop:x-minimum-products}
In the entangled \(w\)-branch, the minimum number of pure-product terms in an optimal separable noise is one when \(\sqrt{bc}<t\le\max\{b,c\}\), and two when \(t>\max\{b,c\}\), where \(t=|w|\). The \(z\)-branch follows by Eq.~\eqref{sm:eq:x-two-product-branch-exchange}.
\end{proposition}
\begin{proof}
Local phase rotations preserve pure-product preparation size, so set \(w=t>0\). Equation~\eqref{sm:eq:x-optimal-two-product-weights-w} supplies a one-term optimum in the first region, including its upper boundary, and a two-term optimum in the second. It remains to exclude any one-term optimum in the strict second region, even one not diagonal in this basis.

There \(R=2w-b-c>0\), and \(W=2\proj{\beta}\), with \(\ket\beta=(\ket{01}-\ket{10})/\sqrt2\), is an optimal generalized-dual certificate. Since
\[
 I-W^\Gamma=2\proj{\Phi^+},\qquad
 \ket{\Phi^+}=(\ket{00}+\ket{11})/\sqrt2,
\]
a hypothetical optimal noise \(Y=R\proj{e\otimes f}\) must obey
\(e_0f_0+e_1f_1=0\). For normalized local vectors, this means \(f=(-e_1,e_0)\) up to phase. Direct partial transposition gives
\begin{equation}
 Y^\Gamma\ket\beta=\frac R{\sqrt2}
 \begin{pmatrix}-e_0\overline e_1\\|e_0|^2\\-|e_1|^2\\e_1\overline e_0\end{pmatrix},
 \qquad
 A\ket\beta=\frac1{\sqrt2}
 \begin{pmatrix}0\\b-w\\w-c\\0\end{pmatrix}.
 \label{sm:eq:x-single-product-obstruction}
\end{equation}
The second complementary-slackness condition,
\((A+Y^\Gamma)\ket\beta=0\), forces \(e_0\overline e_1=0\). Thus the only possible product projectors are \(\proj{01}\) and \(\proj{10}\). The former requires \(w=c\), and the latter \(w=b\), contradicting \(w>\max\{b,c\}\). Two terms are therefore necessary and sufficient. This statement minimizes preparation size over all separable optimal noises; it does not assert uniqueness of their decompositions.
\end{proof}

\begin{figure}[t]
 \centering
 \includegraphics[width=0.94\linewidth]{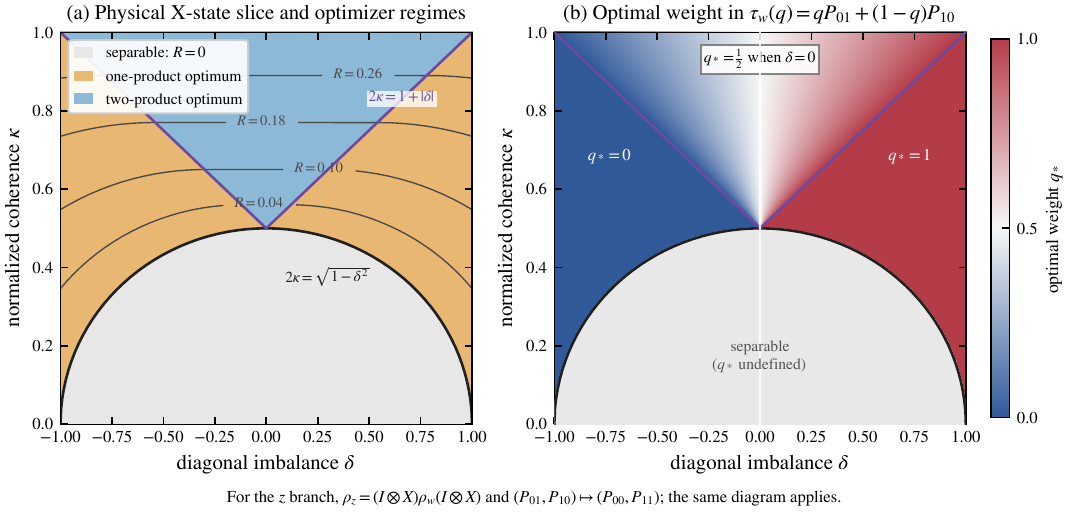}
 \caption{Minimum pure-product preparation size of an optimal separable noise. For the physical X-state slice \(a=d=1/3\), \(b=(1-\delta)/6\), \(c=(1+\delta)/6\), \(w=\kappa/3\), \(z=0\), with \(-1\le\delta\le1\), \(0\le\kappa\le1\), the PPT boundary is \(2\kappa=\sqrt{1-\delta^2}\). Within the entangled region, \(2\kappa=1+|\delta|\) separates one-term and two-term optimal preparation; the boundary admits one term. In panel (b), the color gives \(q_*=u_*/(u_*+v_*)\) for the representative noise in Eq.~\eqref{sm:eq:x-optimal-two-product-weights-w}. Here \(P_{ij}=\proj{ij}\) and \(\tau_w(q)=qP_{01}+(1-q)P_{10}\). Equal weights occur on \(\delta=0\) in the two-term region.}
 \label{sm:fig:xstate-phase-weight}
\end{figure}

Figure~\ref{sm:fig:xstate-phase-weight} displays the minimum-preparation regions and the representative weight \(q_*=u_*/(u_*+v_*)\). We next compare these weights with those for the pure BSA remainder.

Write the unique BSA of an entangled two-qubit state as
\begin{equation}
 \rho_X=S+\varepsilon\proj{\psi},
 \qquad S\in\SEP_+,
 \qquad\varepsilon>0,
 \label{sm:eq:x-bsa-decomposition}
\end{equation}
where \(\Tr S\) is maximal. A pure entangled remainder and uniqueness of the optimal two-qubit decomposition were established in Refs.~\cite{LewensteinSanpera1998,KarnasLewenstein2001}.
\begin{proposition}[Parity and robustness noise of the BSA remainder]
\label{sm:prop:x-bsa-remainder}
If \(\rho_X\) lies in the \(w\)-entangled branch, then
\begin{equation}
 \ket{\psi}=\alpha\ket{00}+e^{i\phi}\beta\ket{11},
 \qquad
 \alpha,\beta>0,
 \quad
 \alpha^2+\beta^2=1.
 \label{sm:eq:x-bsa-even-remainder}
\end{equation}
One optimal separable noise for \(\proj{\psi}\) is
\begin{equation}
 \omega_{\psi}=\frac12\proj{01}+\frac12\proj{10},
 \qquad
 \Rs(\proj{\psi})=\Rg(\proj{\psi})=2\alpha\beta.
 \label{sm:eq:x-bsa-equal-noise}
\end{equation}
In the \(z\)-entangled branch, the remainder has odd parity and the product pair is replaced by \(\{\ket{00},\ket{11}\}\).
\end{proposition}

\begin{proof}
Conjugating Eq.~\eqref{sm:eq:x-bsa-decomposition} by \(\Pi\) produces another BSA with the same separable weight. Uniqueness therefore gives \(\Pi\proj{\psi}\Pi=\proj{\psi}\), so \(\ket{\psi}\) has definite parity. In the \(w\)-entangled branch, the odd-parity principal block of \(\rho_X^\Gamma\) is not positive semidefinite. If \(\ket{\psi}\) had odd parity, the same block of \((\proj{\psi})^\Gamma\) would be diagonal and positive semidefinite. The corresponding block of \(S^\Gamma\) is also positive semidefinite because \(S\) is separable. Their sum could not equal the nonpositive odd-parity block of \(\rho_X^\Gamma\). Hence the remainder has the even-parity form in Eq.~\eqref{sm:eq:x-bsa-even-remainder}; exchanging parity sectors gives the \(z\)-branch.

For the state in Eq.~\eqref{sm:eq:x-bsa-even-remainder}, add the unnormalized noise
\begin{equation}
 2\alpha\beta\,\omega_{\psi}
 =\alpha\beta\proj{01}+\alpha\beta\proj{10}.
\end{equation}
The only nontrivial block of the partial transpose of the resulting operator is
\begin{equation}
 \alpha\beta
 \begin{pmatrix}
  1&e^{-i\phi}\\
  e^{i\phi}&1
 \end{pmatrix}
 \succeq0.
\end{equation}
The mixture is therefore PPT and separable. Conversely, \((\proj{\psi})^\Gamma\) has eigenvalue \(-\alpha\beta\) with maximally entangled eigenvector
\begin{equation}
 \frac{\ket{01}-e^{i\phi}\ket{10}}{\sqrt2}.
\end{equation}
The single-vector lower bound in Eq.~\eqref{sm:eq:single-vector-formula} forces every physical noise to have weight at least \(2\alpha\beta\). This proves Eq.~\eqref{sm:eq:x-bsa-equal-noise}.
\end{proof}

The BSA fixes the parity of its pure remainder and hence the relevant cross-Schmidt product pair, but it does not fix the mixed-state robustness weights. In the \(w\)-branch, both the pure BSA remainder and the original mixed X state admit optimal noises supported on \(\{\ket{01},\ket{10}\}\). The pure remainder uses equal weights, whereas the mixed state uses
\begin{equation}
 q_*:=\frac{u_*}{u_*+v_*},
 \qquad
 \omega_*=q_*\proj{01}+(1-q_*)\proj{10},
 \label{sm:eq:x-bsa-versus-mixed-weight}
\end{equation}
which is generally asymmetric and may reduce to a single product state. The \(z\)-branch uses \(\{\ket{00},\ket{11}\}\).

\begin{figure}[t]
 \centering
 \includegraphics[width=0.94\linewidth]{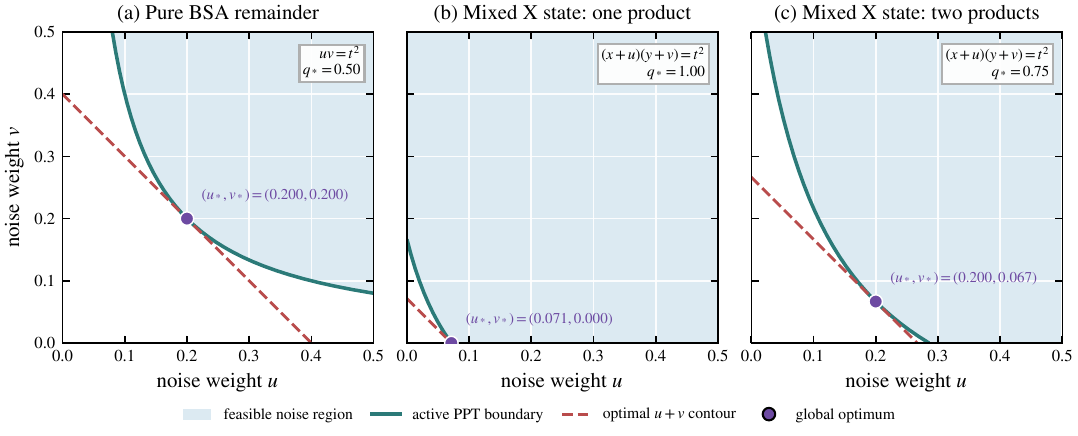}
 \caption{Representative crossed-product optimal noises for separate inputs: (a) a pure BSA remainder with \(\alpha\beta=1/5\); (b,c) mixed X states with \(a=d=1/3\), \(b=1/10\), \(c=7/30\), \(z=0\), and \(w=1/5,3/10\), respectively. The weights multiply \(\proj{01}\) and \(\proj{10}\). In the plots, \((x,y,t)=(0,0,\alpha\beta)\) for (a) and \((b,c,w)\) for (b,c). The shifted constraint \((x+u)(y+v)\ge t^2\) produces equal, one-sided, or unequal weights. Panel (a) is not the BSA of either mixed input.}
 \label{sm:fig:bsa-xstate-geometry}
\end{figure}

\section{Proof of the negativity bounds}\label{sm:sec:bounds}

For a two-qubit state, use the negativity convention
\begin{equation}
 \Neg(\rho):=\frac{\|\rho^\Gamma\|_1-1}{2}.
 \label{sm:eq:negativity-definition}
\end{equation}
If \(\rho\) is entangled, let \(\ket{v_-}\) be the normalized eigenvector of \(\rho^\Gamma\) with eigenvalue \(-\Neg(\rho)\), and let \(C_-\) be its pure-state concurrence. Denote Wootters' mixed-state concurrence by \(C(\rho)\).

\begin{proposition}[Negativity bounds and equality condition]
\label{sm:prop:negativity-bounds}
For every entangled two-qubit state,
\begin{equation}
 \frac{2\Neg(\rho)}{1+\sqrt{1-C_-^2}}
 \le \Rob(\rho)\le2\Neg(\rho)\le C(\rho).
 \label{sm:eq:negativity-robustness-bounds}
\end{equation}
Moreover, \(\Neg(\rho)<\Rob(\rho)\), and
\begin{equation}
 \Rob(\rho)=2\Neg(\rho)
 \quad\Longleftrightarrow\quad
 \ket{v_-}\text{ is maximally entangled}.
 \label{sm:eq:upper-bound-equality}
\end{equation}
\end{proposition}

\begin{proof}
The vector \(\ket{v_-}\) cannot be a product vector, because Eq.~\eqref{sm:eq:A-product} makes every product expectation of \(\rho^\Gamma\) nonnegative. Write its Schmidt coefficients as \(\alpha_1\ge\alpha_2>0\). Choosing \(\ket b=\ket{v_-}\) in Eq.~\eqref{sm:eq:single-vector-formula} gives
\begin{equation}
 \Rob(\rho)\ge\frac{\Neg(\rho)}{\alpha_1^2}
 =\frac{2\Neg(\rho)}{1+\sqrt{1-C_-^2}},
 \label{sm:eq:negative-vector-lower}
\end{equation}
where \(C_-=2\alpha_1\alpha_2\). Because \(\alpha_1<1\), this lower bound is strictly larger than \(\Neg(\rho)\).

Choose three normalized eigenvectors \(\ket{v_1},\ket{v_2},\ket{v_3}\) so that, together with \(\ket{v_-}\), they form an orthonormal eigenbasis of \(\rho^\Gamma\). Write their nonnegative eigenvalues as \(\eta_1,\eta_2,\eta_3\). For every normalized \(\ket b\),
\begin{align}
 -\bra b\rho^\Gamma\ket b
 &=\Neg(\rho)|\braket{v_-}{b}|^2
   -\sum_{j=1}^3\eta_j|\braket{v_j}{b}|^2\notag\\
 &\le\Neg(\rho).
 \label{sm:eq:numerator-negativity-bound}
\end{align}
Since every two-qubit vector has \(\alpha_1^2(b)\ge1/2\), Eq.~\eqref{sm:eq:single-vector-formula} proves \(\Rob(\rho)\le2\Neg(\rho)\). The established comparison between doubled negativity and concurrence gives \(2\Neg(\rho)\le C(\rho)\) \cite{VerstraeteEtAl2001}. This also follows from convexity of negativity: every pure two-qubit state satisfies \(2\Neg=C\), and applying convexity to a concurrence-minimizing pure-state ensemble \cite{Wootters1998} gives the mixed-state inequality.

If \(\ket{v_-}\) is maximally entangled, choosing it in Eq.~\eqref{sm:eq:negative-vector-lower} attains \(2\Neg(\rho)\). Conversely, suppose the upper bound is attained and let \(b_*\) maximize Eq.~\eqref{sm:eq:single-vector-formula}. Equality in
\begin{equation}
 \frac{-\bra{b_*}\rho^\Gamma\ket{b_*}}
      {\alpha_1^2(b_*)}
 \le\frac{\Neg(\rho)}{\alpha_1^2(b_*)}
 \le2\Neg(\rho)
\end{equation}
requires \(\alpha_1^2(b_*)=1/2\) and equality in Eq.~\eqref{sm:eq:numerator-negativity-bound}. The latter forces \(|\braket{v_-}{b_*}|=1\), so \(v_-\) itself is maximally entangled.
\end{proof}

The eigenvector lower bound is not a pointwise tight lower frontier except at \(C_-=1\). Indeed, write
\(\ket{v_-}=\cos\theta\ket{u_1f_1}+\sin\theta\ket{u_2f_2}\)
with \(0<\theta<\pi/4\), and increase \(\theta\) while fixing the Schmidt bases. Since \(v_-\) is an eigenvector of \(A\), the derivative of \(-\bra vA\ket v\) vanishes there, whereas
\(d(\cos^2\theta)/d\theta=-\sin(2\theta)<0\). The variational quotient therefore increases to first order, proving strict inequality in Eq.~\eqref{sm:eq:negative-vector-lower} for \(0<C_-<1\). At \(C_-=1\), the lower and upper bounds coincide. No claim of global attainability of either plotted lower curve at every parameter value is made.

These lower bounds are the negativity-based teleportation bounds of
Ref.~\cite{VerstraeteVerschelde2003}, expressed in the convention
\(\Neg=(\|\rho^\Gamma\|_1-1)/2\) and now applying also to \(\Rs\).
That reference uses doubled negativity and establishes
\(C_-\ge2\Neg(\rho)/C(\rho)\). Since the lower bound increases with \(C_-\),
\begin{equation}
 R(\rho)\ge
 \frac{2\Neg(\rho)}
 {1+\sqrt{1-[2\Neg(\rho)/C(\rho)]^2}},
 \label{sm:eq:mixed-state-NC-lower}
\end{equation}
which gives panel (b) of Fig.~\ref{sm:fig:robustness-negativity-concurrence}.
The equality condition in Eq.~\eqref{sm:eq:upper-bound-equality} follows from
the two simultaneous equalities in the proof above. Pure states and
Bell-diagonal states recover \(R=2\Neg=C\), consistently with their
established robustness values \cite{VidalTarrach1999,AkhtarshenasJafarizadeh2003}.

\begin{figure}[t]
 \centering
 \includegraphics[width=0.96\linewidth]{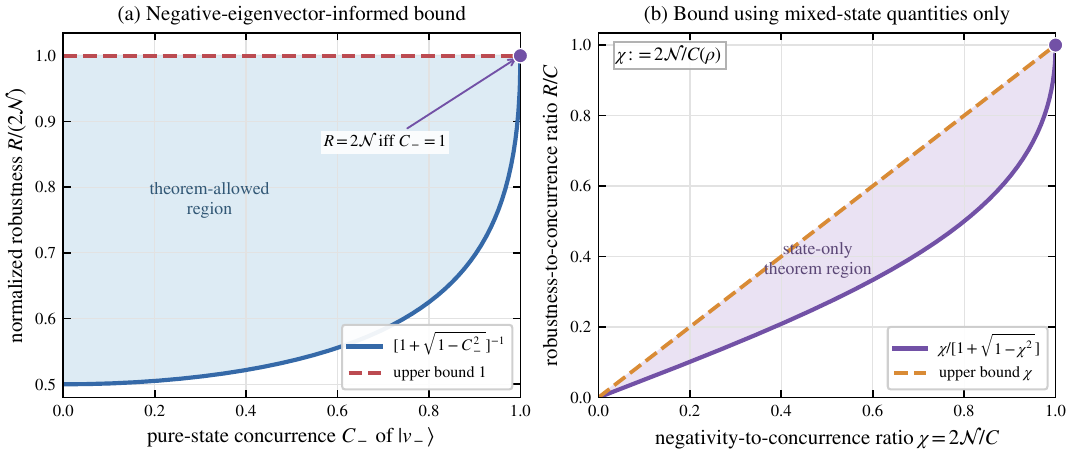}
 \caption{Bounds on the common two-qubit robustness \(\Rob=\Rs=\Rg\). (a) Retaining the pure-state concurrence \(C_-\) of the negative eigenvector of \(\rho^\Gamma\) gives \([1+\sqrt{1-C_-^2}]^{-1}\leq\Rob/(2\Neg)\leq1\). (b) Writing \(\chi=2\Neg/C(\rho)\) gives \(\chi/[1+\sqrt{1-\chi^2}]\leq\Rob/C(\rho)\leq\chi\). Shading represents the proved inequalities, not an attainability claim for every interior point.}
 \label{sm:fig:robustness-negativity-concurrence}
\end{figure}

Figure~\ref{sm:fig:robustness-negativity-concurrence} compares the eigenvector-sensitive lower bound with its coarser state-only consequence.

\section{Operational corollaries and task normalization}\label{sm:sec:operations}
These corollaries use the operational theorems of
Ref.~\cite{TakagiRegula2019}.
The free inputs are separable states on the original bipartition, while
the channel ensembles and output measurements in these tasks are
unrestricted quantum operations and measurements. No additional
entangled reference input is supplied.

For a finite channel ensemble \(\mathcal E=\{p_i,\Lambda_i\}\) and output
POVM \(M=\{M_i\}\), write
\begin{equation}
 p_{\rm succ}(\mathcal E,M,\rho)
 =\sum_i p_i\Tr[M_i\Lambda_i(\rho)].
\end{equation}
Theorem~1 of Ref.~\cite{TakagiRegula2019} gives
\begin{equation}
 \mathcal A_{\rm gen}(\rho):=
 \sup_{\mathcal E,M}
 \frac{p_{\rm succ}(\mathcal E,M,\rho)}
 {\max_{\sigma\in\SEP}p_{\rm succ}(\mathcal E,M,\sigma)}
 =1+\Rg(\rho).
 \label{sm:eq:operational-general}
\end{equation}
The measurement is held fixed between numerator and denominator before
the outer optimization. For equiprobable binary channels, instead define
the \emph{optimized gain above random guessing}
\begin{equation}
 g(\Lambda_0,\Lambda_1;\rho)
 :=\max_M p_{\rm succ}(\{1/2,\Lambda_i\}_{i=0}^1,M,\rho)-\frac12.
\end{equation}
Theorem~7 of the same reference yields
\begin{equation}
 \mathcal A_{\rm bin}(\rho):=
 \sup_{\Lambda_0,\Lambda_1}
 \frac{g(\Lambda_0,\Lambda_1;\rho)}
 {\max_{\sigma\in\SEP}g(\Lambda_0,\Lambda_1;\sigma)}
 =1+2\Rs(\rho).
 \label{sm:eq:operational-binary}
\end{equation}
Only positive denominators are included. Since separable states span the
Hermitian space, zero denominator implies identical channels and zero
gain for every input. Thus \(\mathcal A_{\rm bin}\) is a ratio of
\emph{gains}, not of success probabilities. Theorem~\ref{sm:thm:equality}
gives \(\mathcal A_{\rm bin}=2\mathcal A_{\rm gen}-1\) for two qubits,
without identifying the two discrimination tasks.

For optimal deterministic two-qubit teleportation,
Refs.~\cite{VerstraeteVerschelde2003,HorodeckiTeleportation1999} give
\(f_{\rm tel}^*=(2F_{\rm TP\mbox{-}LOCC}^*+1)/3\).
Here \(F_{\rm TP\mbox{-}LOCC}^*\) is the optimized singlet fraction under
trace-preserving local operations and classical communication (TP-LOCC),
and \(f_{\rm tel}^*\) is the corresponding mean teleportation fidelity.
Together with
Eq.~\eqref{sm:eq:vv-rg-normalization} and Theorem~\ref{sm:thm:equality},
\begin{equation}
 R=2F_{\rm TP\mbox{-}LOCC}^*-1=3f_{\rm tel}^*-2.
 \label{sm:eq:teleportation-common-R}
\end{equation}
Here \(R=\Rs(\rho)=\Rg(\rho)\); the equality identifies the teleportation expression with the standard separable-noise cost.

\section{Exact qubit--qutrit separation}\label{sm:sec:dimension}
\subsection{Counterexample and strict gap}
The proof is specific to \(2\otimes2\).  Its rank reduction relies on the fact that every subspace of dimension at least two contains a product vector.  In \(2\otimes3\), PPT and separability still coincide, but two-dimensional completely entangled subspaces exist, so the reduction can stop at rank two and Proposition~\ref{sm:prop:rank-one-bridge} cannot be invoked.  The equality itself can fail already in \(2\otimes3\), as the following example shows.

Here is an exact counterexample.  In the ordered basis \(\{\ket{00},\ket{01},\ket{02},\ket{10},\ket{11},\ket{12}\}\), let
\begin{equation}
 \rho_*=\frac1{40}
 \begin{pmatrix}
 2&0&0&0&0&0\\
 0&1&0&-4&0&0\\
 0&0&13&0&-5&0\\
 0&-4&0&17&0&0\\
 0&0&-5&0&2&0\\
 0&0&0&0&0&5
 \end{pmatrix}.
 \label{sm:eq:2x3-counterexample}
\end{equation}
Throughout this section, \(\Gamma\) denotes partial transpose on the qutrit factor. The trace of \(\rho_*\) is one. Up to a simultaneous permutation of rows and columns, \(40\rho_*\) has scalar blocks \(2,5\) and the two blocks
\begin{equation}
 \begin{pmatrix}1&-4\\-4&17\end{pmatrix},
 \qquad
 \begin{pmatrix}13&-5\\-5&2\end{pmatrix}.
 \label{sm:eq:2x3-positive-blocks}
\end{equation}
Both have positive diagonal entries and determinant one. Thus \(\rho_*\succ0\): the counterexample is full rank without a perturbation.

Because PPT is equivalent to separability in \(2\otimes3\) \cite{Horodecki1996}, the same primal and dual cone formulations as above, now with \(6\times6\) variables, are exact.  Here \(R_{\mathrm g}^{\mathrm{ent}}\) and \(R_{\mathrm s}^{\mathrm{ent}}\) denote the definitions \eqref{sm:eq:def-rg} and \eqref{sm:eq:def-rs}, respectively, with separability taken across the qubit--qutrit bipartition. The explicit noise and witness below give
\begin{equation}
 R_{\mathrm g}^{\mathrm{ent}}(\rho_*)\le\frac{1}{8}
 <\frac{2}{15}\le R_{\mathrm s}^{\mathrm{ent}}(\rho_*).
 \label{sm:eq:2x3-bounds}
\end{equation}
Thus
\begin{equation}
 R_{\mathrm s}^{\mathrm{ent}}(\rho_*)
 -R_{\mathrm g}^{\mathrm{ent}}(\rho_*)
 \ge\frac{1}{120}>0.
 \label{sm:eq:2x3-strict-gap}
\end{equation}
These are certified bounds, not asserted exact optimal values. Subsection~\ref{sm:app:2x3-certificate} proves feasibility analytically, including the witness bound for arbitrary complex product vectors.

For bipartite systems with local dimensions \(m,n\ge2\) and \(mn>6\), PPT is not sufficient for separability in general.  Consequently, the PPT semidefinite programs used here do not, without further constraints, represent the two entanglement robustnesses in those dimensions.  Explicit PPT-entangled states already occur in \(3\otimes3\) and \(2\otimes4\) \cite{HorodeckiPPT1997}.

\subsection{Exact rational certificate}\label{sm:app:2x3-certificate}
Define the unnormalized vector and weighted noise
\begin{equation}
 \ket{u}=2\ket{01}+\ket{10},\qquad
 Y_g=\frac{\proj{u}}{40},\qquad S=\rho_*+Y_g.
 \label{sm:eq:2x3-noise}
\end{equation}
Then \(Y_g\succeq0\), \(\Tr Y_g=1/8\), and \(S\succ0\). In the reordered basis \(00,11;01,12;02;10\), the partially transposed target is
\begin{equation}
 40S^\Gamma\simeq
 \begin{pmatrix}2&-2\\-2&2\end{pmatrix}
 \oplus\begin{pmatrix}5&-5\\-5&5\end{pmatrix}
 \oplus(13)\oplus(18)\succeq0,
 \label{sm:eq:2x3-target-blocks}
\end{equation}
where \(\simeq\) denotes a basis permutation. The normalized target \(S/(9/8)\) is therefore separable, proving \(\Rg(\rho_*)\le1/8\). The normalized noise \(\proj{u}/5\) is an entangled pure state; its optimality is not assumed.

For the standard lower bound, define unnormalized vectors and Hermitian operators
\begin{equation}
 \begin{gathered}
 \ket{v_1}=\ket{00}+\ket{11},\qquad
 \ket{v_2}=\ket{01}+\ket{12},\\
 W=\tfrac23(\proj{v_1}+\proj{v_2}),\qquad Z=W^\Gamma .
 \end{gathered}
 \label{sm:eq:2x3-witness}
\end{equation}
Let \(a=(a_0,a_1)^{\mathsf T}\) and \(b=(b_0,b_1,b_2)^{\mathsf T}\) be arbitrary unit complex vectors. With
\begin{equation}
 T_a=\begin{pmatrix}a_0&a_1&0\\0&a_0&a_1\end{pmatrix},
 \qquad
 T_aT_a^\dagger=
 \begin{pmatrix}1&a_1\overline{a_0}\\a_0\overline{a_1}&1\end{pmatrix},
 \label{sm:eq:2x3-product-map}
\end{equation}
the induced Euclidean operator norm satisfies \(\|T_a\|^2=1+|a_0a_1|\le3/2\), and hence
\begin{equation}
 0\le\bra{a\otimes b}W\ket{a\otimes b}
 =\tfrac23\|T_ab\|^2\le1.
 \label{sm:eq:2x3-product-bound}
\end{equation}
Partial transpose replaces \(b\) by \(\overline b\) in this expectation, so the same bound holds for \(Z\). By convexity, \(0\le\Tr(Z\sigma)\le1\) for every normalized separable \(\sigma\). For any standard-feasible identity \(\rho_*+t\tau=(1+t)\sigma\), with \(\sigma,\tau\) separable, this implies
\begin{equation}
 -\Tr(Z\rho_*)=t\Tr(Z\tau)-(1+t)\Tr(Z\sigma)\le t.
 \label{sm:eq:2x3-weak-duality}
\end{equation}
Direct multiplication gives \(WS^\Gamma=0\) and
\begin{equation}
 -\Tr(Z\rho_*)=\Tr(ZY_g)
 =\frac{\bra{u}Z\ket{u}}{40}=\frac{2}{15}.
 \label{sm:eq:2x3-witness-value}
\end{equation}
This proves the lower bound and the strict gap in \eqref{sm:eq:2x3-bounds}--\eqref{sm:eq:2x3-strict-gap}, with no floating-point positivity test or assumption on the reality of product vectors.

The range of \(W\) is the completely entangled subspace
\(E=\operatorname{span}\{v_1,v_2\}\) used in the main text: its coefficient matrix
\(\left(\begin{smallmatrix}x&y&0\\0&x&y\end{smallmatrix}\right)\)
has minors \(x^2,y^2\), so a nonzero vector in \(E\) cannot be a product vector. Although \(Z\) is at most one on separable states,
\begin{equation}
 \lambda_{\max}(Z)=\frac{1+\sqrt5}{3}>1,\qquad
 \frac{\bra{u}Z\ket{u}}{\braket{u}{u}}=\frac{16}{15}>1.
 \label{sm:eq:2x3-witness-excess}
\end{equation}
It therefore violates the generalized-dual constraint \(Z\preceq I_6\). A difference between dual feasible sets alone would not establish separation for a physical input. Here \(WS^\Gamma=0\) and \(\rho_*=S-Y_g\succ0\) close that argument:
\begin{equation}
 -\Tr(Z\rho_*)-\Tr Y_g=\Tr[(Z-I_6)Y_g]=\frac{1}{120}.
 \label{sm:eq:2x3-gap-mechanism}
\end{equation}

The construction belongs to a parameter family. Let \(M\) be real symmetric with diagonal \((a,b,c,d,a,b)\), off-diagonal entries \(M_{01,10}=M_{10,01}=-a\) and \(M_{02,11}=M_{11,02}=-b\), and all other entries zero. For real parameters satisfying
\begin{equation}
 a>0,\qquad b>4,\qquad ac>b^2,\qquad
 (b-4)(d-1)>(a+2)^2,
 \label{sm:eq:2x3-family-conditions}
\end{equation}
the matrix \(M-\proj{u}\) is positive definite: its nontrivial blocks are
\(\left(\begin{smallmatrix}b-4&-a-2\\-a-2&d-1\end{smallmatrix}\right)\)
and \(\left(\begin{smallmatrix}c&-b\\-b&a\end{smallmatrix}\right)\), and its scalar blocks are \(a,b\).
Put \(T=\Tr(M-\proj{u})=2a+2b+c+d-5>0\) and
\(\rho(a,b,c,d)=(M-\proj{u})/T\).
The target \(M/T\) is positive and PPT, since \(M^\Gamma\) has blocks
\(\left(\begin{smallmatrix}a&-a\\-a&a\end{smallmatrix}\right)\),
\(\left(\begin{smallmatrix}b&-b\\-b&b\end{smallmatrix}\right)\), \(c,d\).
Using the noise \(\proj{u}/T\) and the same \(Z\), with \(WM^\Gamma=0\), gives
\begin{equation}
 \Rg(\rho(a,b,c,d))\le\frac5T,\qquad
 \Rs(\rho(a,b,c,d))\ge\frac{16}{3T},\qquad
 \Rs(\rho(a,b,c,d))-\Rg(\rho(a,b,c,d))\ge\frac1{3T}.
 \label{sm:eq:2x3-family-bounds}
\end{equation}
The choice \((a,b,c,d)=(2,5,13,18)\) yields \(\rho_*\) and \(T=40\).

\subsection{Full-rank stability}
The state \(\rho_*\) is already full rank. To quantify stability under additional white noise, put \(U=1/8\), \(L=2/15\), and \(g=L-U=1/120\). For \(0\le\epsilon\le1\), define
\begin{equation}
 \rho_\epsilon=(1-\epsilon)\rho_*+\epsilon I_6/6,\qquad
 Y_\epsilon=(1-\epsilon)Y_g .
 \label{sm:eq:2x3-white-noise-family}
\end{equation}
Both \(Y_\epsilon\succeq0\) and
\((\rho_\epsilon+Y_\epsilon)^\Gamma
=(1-\epsilon)S^\Gamma+\epsilon I_6/6\succeq0\).
The target is separable by the \(2\otimes3\) PPT criterion, so
\(\Rg(\rho_\epsilon)\le(1-\epsilon)U\).
The same witness \(Z\) remains feasible independently of the input. Since
\begin{equation}
 \Tr Z=\Tr W=\frac83,
 \label{sm:eq:2x3-witness-trace}
\end{equation}
weak duality gives
\begin{equation}
 \Rs(\rho_\epsilon)-\Rg(\rho_\epsilon)
 \ge (1-\epsilon)g-\frac{\epsilon}{6}\Tr Z
 =\frac{3-163\epsilon}{360}.
 \label{sm:eq:2x3-stability-bound}
\end{equation}
This is positive for \(0\le\epsilon<3/163\). In particular, \(\rho_{1/200}\succeq I_6/1200\) is full rank and
\begin{equation}
 \Rs(\rho_{1/200})-\Rg(\rho_{1/200})
 \ge\frac{437}{72000}>0.00606.
 \label{sm:eq:2x3-fullrank-gap}
\end{equation}
This is an exact lower bound, not a numerical estimate of the optimal gap.

\subsection{Local isometry invariance and the complete dimension boundary}
\begin{proposition}[Local isometry invariance]
\label{sm:prop:isometry-invariance}
Let \(V_A:\mathbb C^{d_A}\to\mathbb C^{D_A}\) and
\(V_B:\mathbb C^{d_B}\to\mathbb C^{D_B}\) be isometries, and put
\(V=V_A\otimes V_B\). For either robustness defined with the actual separable cone,
\begin{equation}
 R_\nu(V\rho V^\dagger)=R_\nu(\rho),\qquad \nu\in\{\mathrm s,\mathrm g\}.
 \label{sm:eq:isometry-invariance}
\end{equation}
\end{proposition}
\begin{proof}
Embedding any feasible noise and target by \(V\) preserves positivity, separability, and trace, so the left-hand side cannot exceed the right-hand side. For the reverse inequality, choose arbitrary local density operators \(\tau_A,\tau_B\). Each map
\begin{equation}
 \mathcal R_j(X)=V_j^\dagger X V_j+
 \Tr[(I-V_jV_j^\dagger)X]\tau_j,\qquad j=A,B,
 \label{sm:eq:local-retraction}
\end{equation}
is completely positive and trace preserving. Its two terms are a compression and a measure-and-prepare map, and
\(\mathcal R_j(V_jXV_j^\dagger)=X\).
Applying \(\mathcal R_A\otimes\mathcal R_B\) to any feasible embedded noise and target preserves their traces and the required cones, and recovers \(\rho\) as the input. It therefore gives a feasible original noise with the same cost, proving the other inequality. No PPT relaxation in the larger dimensions is used.
\end{proof}

For finite integers \(m,n\ge2\), Theorem~\ref{sm:thm:equality},
Eq.~\eqref{sm:eq:2x3-strict-gap}, and Proposition~\ref{sm:prop:isometry-invariance} imply
\begin{equation}
 \left[\Rs(\rho)=\Rg(\rho)\ 
 \text{for every }\rho\in\mathcal D(\mathbb C^m\otimes\mathbb C^n)\right]
 \quad\Longleftrightarrow\quad (m,n)=(2,2).
 \label{sm:eq:all-dimensions-classification}
\end{equation}
Indeed, every other such pair contains either \(2\otimes3\) or \(3\otimes2\); exchanging subsystems leaves both robustnesses unchanged. The embedded \(\rho_*\) retains its strictly positive gap. This is an existential separation in every larger nontrivial bipartite dimension, not an assertion of separation for every state.

\subsection{Any fixed bipartition of a multipartite system}
\begin{corollary}[Complete finite-dimensional classification across a fixed cut]
\label{sm:cor:fixed-bipartition}
Let \(\mathcal H=\bigotimes_{i=1}^k\mathbb C^{d_i}\), with finite integers \(d_i\ge1\), and fix a nonempty proper subset \(S\) of the parties. Put
\begin{equation}
 d_S=\prod_{i\in S}d_i,\qquad d_{\bar S}=\prod_{i\notin S}d_i.
 \label{sm:eq:effective-cut-dimensions}
\end{equation}
Define \(R_{\mathrm s}^{S|\bar S}\) and \(R_{\mathrm g}^{S|\bar S}\) using separability across this fixed cut. Then
\begin{equation}
 \left[R_{\mathrm s}^{S|\bar S}(\rho)=R_{\mathrm g}^{S|\bar S}(\rho)
       \ \text{for every state }\rho\text{ on }\mathcal H\right]
 \quad\Longleftrightarrow\quad
 \left[\min\{d_S,d_{\bar S}\}=1\ \text{or }(d_S,d_{\bar S})=(2,2)\right].
 \label{sm:eq:fixed-cut-classification}
\end{equation}
\end{corollary}
\begin{proof}
Grouping the factors in \(S\) and \(\bar S\) identifies \(\mathcal H\) with
\(\mathbb C^{d_S}\otimes\mathbb C^{d_{\bar S}}\). Under this identification, the free cone is exactly
\[
 \SEP_+^{S|\bar S}
 =\left\{\sum_r X_r^{S}\otimes X_r^{\bar S}:X_r^S,X_r^{\bar S}\succeq0\right\}.
\]
Thus both optimizations become the bipartite definitions with these effective dimensions. If either dimension is one, every state is separable and both robustnesses vanish. Otherwise Eq.~\eqref{sm:eq:all-dimensions-classification} applies. In particular, in each larger nontrivial effective dimension, the local embeddings of Proposition~\ref{sm:prop:isometry-invariance} place a separating state in \(\mathcal H\).
\end{proof}

This completes the dimension classification of universal equality for finite-dimensional bipartite systems, and applies to any fixed bipartition of a multipartite system. Here each grouped side may contain internal entanglement. Separability across a cut is not full separability among all parties, nor a union or mixture of different cuts. The corollary is applied to each chosen cut with its own free cone; it makes no assertion that arbitrary cuts satisfy universal equality. It also says nothing about tensor powers, smoothing, or asymptotic operational tasks.

\section{Full separability: multipartite classification and GHZ certificates}
\label{sm:sec:fs-multipartite}
This section concerns a different free set from the fixed-cut cone of
Corollary~\ref{sm:cor:fixed-bipartition}.  Let
\(
 \mathcal H_N=\bigotimes_{j=1}^N\mathbb C^{d_j}
\)
with finite \(d_j\ge1\), and define
\begin{equation}
 \FS_N=\operatorname{conv}\left\{
  \bigotimes_{j=1}^N\proj{x_j}:\ket{x_j}\in\mathbb C^{d_j},\
  \braket{x_j}{x_j}=1
 \right\},
 \qquad
 K_N=\{t\sigma:t\ge0,\ \sigma\in\FS_N\}.
 \label{sm:eq:fs-set}
\end{equation}
For a density operator \(\rho\) on \(\mathcal H_N\), set
\begin{align}
 R_{\mathrm s}^{\FS_N}(\rho)
 &=\inf\left\{t\ge0:
   \frac{\rho+t\tau}{1+t}\in\FS_N,\ \tau\in\FS_N\right\},
 \label{sm:eq:fs-rs}
 \\
 R_{\mathrm g}^{\FS_N}(\rho)
 &=\inf\left\{t\ge0:
   \frac{\rho+t\tau}{1+t}\in\FS_N,\ \tau\in\mathcal D(\mathcal H_N)\right\}.
 \label{sm:eq:fs-rg}
\end{align}
Here \(\mathcal D(\mathcal H_N)\) denotes the set of density operators.  The
variable \(t\) is a noise-to-signal ratio, so its associated noise
probability is \(t/(1+t)\).  The notation \(\FS_N\) is used only in this
section and in the main-text classification.

We use the following weak-duality bounds. If a Hermitian \(Z\) obeys
\(0\le\Tr(Z\sigma)\le1\) for every \(\sigma\in\FS_N\), then
\(R_{\mathrm s}^{\FS_N}(\rho)\ge-\Tr(Z\rho)\).
If instead \(Z\) is nonnegative on \(\FS_N\) and \(Z\preceq I\), the
same bound holds for \(R_{\mathrm g}^{\FS_N}\).
Indeed, every feasible identity \(\rho+t\tau=(1+t)\sigma\) gives
\(-\Tr(Z\rho)=t\Tr(Z\tau)-(1+t)\Tr(Z\sigma)\le t\).

\begin{lemma}[Free maps, spectators, and local embeddings]
\label{sm:lem:fs-free-maps}
If a completely positive trace-preserving map \(\Phi\) sends \(\FS_N\) into
\(\FS_M\), then
\(R_\nu^{\FS_M}(\Phi(\rho))\le R_\nu^{\FS_N}(\rho)\) for
\(\nu\in\{\mathrm s,\mathrm g\}\).  Equality holds when \(\Phi\) has a
free CPTP left inverse (a CPTP left inverse that sends \(\FS_M\) into
\(\FS_N\)).  In particular, for any \(\eta\in\FS_L\),
\begin{equation}
 R_\nu^{\FS_{N+L}}(\rho\otimes\eta)=R_\nu^{\FS_N}(\rho),
 \qquad \nu\in\{\mathrm s,\mathrm g\},
 \label{sm:eq:fs-spectators}
\end{equation}
and local isometries preserve both values.
\end{lemma}
\begin{proof}
Applying \(\Phi\) to a feasible identity
\(\rho+t\tau=(1+t)\sigma\) preserves the free target and the
admissible noise set, proving monotonicity.  A free left inverse gives the
reverse inequality.  For \eqref{sm:eq:fs-spectators}, use
\(X\mapsto X\otimes\eta\) and the partial trace over the added parties.
For local isometries \(V_j\), the reverse map is the tensor product of
the trace-preserving maps
\begin{equation}
 \mathcal R_j(X)=V_j^\dagger X V_j+
 \Tr[(I-V_jV_j^\dagger)X]\,\tau_j,
 \label{sm:eq:fs-retraction}
\end{equation}
where \(\tau_j\) is any local state.  Each \(\mathcal R_j\) is the sum of a
compression and a measure-and-prepare branch and therefore is
completely positive; it recovers the embedded input and preserves full
separability.
\end{proof}

\subsection{\texorpdfstring{All-\(n\) GHZ values}{All-n GHZ values}}
Contreras--Tejada, Palazuelos, and de Vicente proved in Supplemental
Material, Sec.~II, Lemmas~1 and~2, that the three-qubit GHZ and \(W\)
states both have standard robustness two with respect to full
separability \cite{Contreras2019}.  In our notation,
\begin{equation}
 R_{\mathrm s}^{\FS}(G_3)=R_{\mathrm s}^{\FS}(W_3)=2,
 \qquad
 \ket{w_3}=\frac{\ket{001}+\ket{010}+\ket{100}}{\sqrt3},
 \quad W_3=\proj{w_3}.
 \label{sm:eq:contreras-three-qubit-standard-values}
\end{equation}
Their Lemma~1 gives a fully separable GHZ-symmetric decomposition and a
matching witness in Eqs.~(12)--(18); Lemma~2 gives a fully separable
decomposition and a matching witness for \(W_3\) in Eqs.~(19)--(38).
After defining generalized robustness in Eq.~(39), the same source
records the known values
\begin{equation}
 R_{\mathrm g}^{\FS}(G_3)=1,\qquad
 R_{\mathrm g}^{\FS}(W_3)=\frac54,
 \label{sm:eq:contreras-three-qubit-generalized-values}
\end{equation}
so both states exhibit strict separation.  The generalized GHZ value
for arbitrary \(n\) also follows from the stabilizer-state result of
Ref.~\cite{Hayashi2008}, Sec.~III~A, Eqs.~(26)--(31), together with the
maximum squared product overlap \(1/2\).  The construction below supplies
explicit primal and dual certificates for both GHZ robustnesses at every
\(n\).

For \(n\ge2\), set \(N=n\), \(d_j=2\) for every party, and put \(D=2^n\).
\begin{align}
 \ket{g_n^\pm}&=\frac{\ket{0^n}\pm\ket{1^n}}{\sqrt2},
 &G_n^\pm&=\proj{g_n^\pm},
 \nonumber\\
 E_n&=\proj{0^n}+\proj{1^n},
 &\Pi_n&=I_D-E_n,
 \nonumber\\
 C_n&=\ket{0^n}\bra{1^n}+\ket{1^n}\bra{0^n},
 &\sigma_\pm&=\frac{I_D\pm C_n}{D}.
 \label{sm:eq:fs-ghz-operators}
\end{align}

\begin{lemma}[Explicit fully separable phase mixtures]
\label{sm:lem:fs-phase-mixtures}
Both \(\sigma_+\) and \(\sigma_-\) belong to \(\FS_N\) for the
\(n\)-qubit partition.
\end{lemma}
\begin{proof}
Let \(\ket{+_\theta}=(\ket0+e^{i\theta}\ket1)/\sqrt2\).  For
\(\chi\in\{0,\pi\}\), choose \(k_1,\ldots,k_{n-1}\in\{0,1,2\}\), set
\(\theta_j=2\pi k_j/3\) for \(j<n\),
\(\theta_n=\chi-\sum_{j<n}\theta_j\), and average the product projectors:
\begin{equation}
 \Omega_\chi=\frac1{3^{n-1}}\sum_{k_1,\ldots,k_{n-1}=0}^{2}
 \bigotimes_{j=1}^n\proj{+_{\theta_j}}.
 \label{sm:eq:fs-phase-twirl}
\end{equation}
For computational strings \(x,y\), the phase average contains the
factor
\(\prod_{j<n}\exp[2\pi i k_j((x_j-y_j)-(x_n-y_n))/3]\).
The integer in each exponent lies in \(\{-2,-1,0,1,2\}\), so the
average vanishes unless it is zero.  Only diagonal entries and the two
all-zero/all-one coherences survive.  Hence
\(\Omega_0=\sigma_+\) and \(\Omega_\pi=\sigma_-\), each as a finite
convex mixture of pure product projectors.
\end{proof}

\begin{lemma}[GHZ witnesses]
\label{sm:lem:fs-ghz-witnesses}
Let \(\beta_n=D/4=2^{n-2}\) and
\(\alpha_n=D/[2(D-2)]\).  The Hermitian operators
\begin{equation}
 Z_{\mathrm s}=\alpha_n\Pi_n-\beta_n C_n,
 \qquad Z_{\mathrm g}=Z_{\mathrm s}/\beta_n
 \label{sm:eq:fs-ghz-witnesses}
\end{equation}
obey \(0\le\Tr(Z_{\mathrm s}\sigma)\le1\) for all
\(\sigma\in\FS_N\), while \(Z_{\mathrm g}\) is nonnegative on
\(\FS_N\) and satisfies \(Z_{\mathrm g}\preceq I_D\).
\end{lemma}
\begin{proof}
For \(\ket v=\bigotimes_j(a_j\ket0+b_j\ket1)\), with
\(|a_j|^2+|b_j|^2=1\) for every \(j\), define
\(A=\prod_j|a_j|^2\), \(B=\prod_j|b_j|^2\),
\(h=\prod_j a_j^*b_j\), and \(r=|h|=\sqrt{AB}\).  If
\(p_x=|\braket{x}{v}|^2\), then \(p_xp_{\bar x}=r^2\), where
\(\bar x\) is the bitwise complement of \(x\).  Pairing all
\(D-2\) strings other than \(0^n,1^n\) gives
\(1-A-B\ge(D-2)r\), while \(A+B\ge2r\) and \(r\le1/D\).  Therefore
\begin{align}
 \bra vZ_{\mathrm s}\ket v
 &\ge[\alpha_n(D-2)-2\beta_n]r=0,\nonumber\\
 \bra vZ_{\mathrm s}\ket v
 &\le\alpha_n(1-2r)+2\beta_nr
 \le\alpha_n+\frac{2(\beta_n-\alpha_n)}D=1.
 \label{sm:eq:fs-ghz-product-bounds}
\end{align}
Convexity extends these bounds to \(\FS_N\).  The spectrum of
\(Z_{\mathrm g}\) is
\begin{equation}
 \operatorname{spec}(Z_{\mathrm g})=
 \{-1,1,\underbrace{2/(D-2),\ldots,2/(D-2)}_{D-2\ \mathrm{times}}\},
 \label{sm:eq:fs-ghz-spectrum}
\end{equation}
so \(Z_{\mathrm g}\preceq I_D\) and the generalized witness condition
follows.
\end{proof}

\begin{theorem}[Exact GHZ robustness]
\label{sm:thm:fs-ghz}
For every \(n\ge2\),
\begin{equation}
 \boxed{R_{\mathrm g}^{\FS_N}(G_n^+)=1,\qquad
 R_{\mathrm s}^{\FS_N}(G_n^+)=2^{n-2}.}
 \label{sm:eq:fs-ghz-values}
\end{equation}
\end{theorem}
\begin{proof}
The identities
\begin{equation}
 G_n^++G_n^-=E_n,
 \qquad
 G_n^++\beta_n\sigma_-
 =\beta_n\sigma_++\frac12E_n
 \label{sm:eq:fs-ghz-primal-identities}
\end{equation}
give generalized and standard feasible costs \(1\) and \(\beta_n\),
respectively: \(E_n\in K_N\) and Lemma~\ref{sm:lem:fs-phase-mixtures}
supplies the two phase mixtures.  Lemma~\ref{sm:lem:fs-ghz-witnesses} gives
the matching lower bounds \(-\Tr(Z_{\mathrm g}G_n^+)=1\) and
\(-\Tr(Z_{\mathrm s}G_n^+)=\beta_n\).
\end{proof}

The full-separability threshold in the following family is known
\cite{DurCiracTarrach1999}, Sec.~3.3, Eq.~(9). With that paper's GHZ
weight \(x=1-p\), its condition \(x\le1/(1+2^{n-1})\) is
\(p\ge D/(D+2)\). We give matching robustness certificates throughout this family.

\begin{theorem}[Full-rank white-noise GHZ family]
\label{sm:thm:fs-noisy-ghz}
For \(0\le p\le1\), let
\begin{equation}
 \rho_{n,p}=(1-p)G_n^++pI_D/D,
 \qquad p_*={D}/{(D+2)}.
 \label{sm:eq:fs-noisy-ghz}
\end{equation}
Then
\begin{align}
 R_{\mathrm g}^{\FS_N}(\rho_{n,p})
 &=\max\{0,1-p-2p/D\},\nonumber\\
 R_{\mathrm s}^{\FS_N}(\rho_{n,p})
 &=\max\{0,D(1-p)/4-p/2\}
 =2^{n-2}R_{\mathrm g}^{\FS_N}(\rho_{n,p}),
 \label{sm:eq:fs-noisy-ghz-values}
\end{align}
and \(\rho_{n,p}\in\FS_N\) exactly when \(p\ge p_*\).
\end{theorem}
\begin{proof}
For \(p\le p_*\), put
\(s=1-p-2p/D\) and \(t=D(1-p)/4-p/2\).  Direct expansion gives
\begin{align}
 \rho_{n,p}+sG_n^-&=p\sigma_++(1-p-p/D)E_n,\nonumber\\
 \rho_{n,p}+t\sigma_-&=(p+t)\sigma_++(1-p)E_n/2.
 \label{sm:eq:fs-noisy-primal}
\end{align}
All coefficients are nonnegative in this interval.  Since
\(\Tr(Z_{\mathrm s}I_D/D)=1/2\) and
\(\Tr(Z_{\mathrm g}I_D/D)=2/D\), the witnesses in
Lemma~\ref{sm:lem:fs-ghz-witnesses} give the same lower bounds.
For \(p\ge p_*\), write \(\kappa=D(1-p)/2\) and use
\begin{equation}
 \rho_{n,p}=\kappa\sigma_++(p/D)E_n+
 \left[p/D-(1-p)/2\right]\Pi_n.
 \label{sm:eq:fs-noisy-free}
\end{equation}
Every term is fully separable with a nonnegative coefficient, proving
the final assertion and the zero branch.  For \(p>0\),
\(\rho_{n,p}\succeq(p/D)I_D\), so the strict gap for \(n\ge3\) and
\(0<p<p_*\) occurs on full-rank states.
\end{proof}

\subsection{Complete finite-dimensional FS classification}
\begin{theorem}[Universal equality for the fully separable free set]
\label{sm:thm:fs-classification}
Delete all one-dimensional parties and let \(k\) be the number of remaining
parties.  For finite local dimensions,
\begin{equation}
 \left[R_{\mathrm s}^{\FS_N}(\rho)=R_{\mathrm g}^{\FS_N}(\rho)
 \ \text{for every }\rho\right]
 \Longleftrightarrow
 \left[k\le1\ \text{or}\ \bigl(k=2\ \text{and both dimensions are }2\bigr)\right].
 \label{sm:eq:fs-classification}
\end{equation}
Every other finite system contains a strict-separation state.
\end{theorem}
\begin{proof}
For \(k\le1\) every state is fully separable.  For two qubits use
Theorem~\ref{sm:thm:equality}.  If \(k=2\) and one local dimension is at
least three, embed the certified \(2\otimes3\) state of
Sec.~\ref{sm:sec:dimension}; Lemma~\ref{sm:lem:fs-free-maps} preserves its
strict gap. If \(k\ge3\), embed the known three-qubit GHZ separating
state \cite{Contreras2019} into three parties and add pure-product
spectators. Lemma~\ref{sm:lem:fs-free-maps} again preserves the gap.
This proves both directions without requiring the all-\(n\) formula.

Theorem~\ref{sm:thm:fs-ghz} strengthens the last construction: embedding
\(G_k^+\) into a two-dimensional subspace of every nontrivial party
gives values \(1\) and \(2^{k-2}\). This pure state is genuinely
\(k\)-partite entangled, having Schmidt rank two across every
nontrivial bipartition. The white-noise constructions in
Sec.~\ref{sm:sec:dimension} and Theorem~\ref{sm:thm:fs-noisy-ghz} give
full-rank separating states on \(2\otimes3\) and on \(k\) qubits,
respectively. Their isometric embeddings into larger spaces need not
be full rank.
\end{proof}

The classification is existential: it does not say that every state in a
larger system separates the two measures.  The classification applies specifically
to full separability among the displayed parties; fixed-bipartition
separability, biseparable free sets, tensor powers, smoothing, and asymptotic
regularization are separate questions.  For \(n\ge3\), the ratio of the two
GHZ costs is \(2^{n-2}\).  The exponentially growing quantity is the
optimized noise-to-signal ratio, not the noise probability.

\clearpage
\twocolumngrid

\end{document}